\RequirePackage{silence}
\documentclass[11pt]{article}
\usepackage[in]{fullpage}
\usepackage[normalem]{ulem}
\usepackage{enumitem}
\usepackage{alglib-tomer}
\usepackage{amsmath}
\usepackage{amssymb}
\usepackage{amsthm}
\usepackage{amsfonts}
\usepackage{thmtools, thm-restate}
\usepackage{xcolor}
\usepackage{tabularx}
\usepackage{xspace}
\usepackage{calc}
\usepackage{comment}
\usepackage{parskip}
\usepackage{mathtools}
\usepackage{bbm}
\usepackage{hyperref}
\hypersetup{
    colorlinks,
    citecolor=blue,
    filecolor=black,
    linkcolor=black,
    urlcolor=black,
    linktoc=all
}

\newtheorem{lemma}{Lemma}[section]
\newtheorem{theorem}[lemma]{Theorem}
\newtheorem{proposition}[lemma]{Proposition}
\newtheorem{corollary}[lemma]{Corollary}
\newtheorem{observation}[lemma]{Observation}

\theoremstyle{definition}\newtheorem{definition}[lemma]{Definition}

\newcommand{\reject}[0]{\textsc{reject}\xspace}

\newcommand{\abs}[1]{\left|{#1}\right|}
\newcommand{\cond}{\middle |}
\newcommand{\floor}[1]{{\left\lfloor{#1}\right\rfloor}}
\newcommand{\ceil}[1]{{\left\lceil{#1}\right\rceil}}

\newcommand{\E}{\mathop{{\rm E}\/}}
\newcommand{\Var}{\mathop{{\rm Var}\/}}
\newcommand{\eps}{\varepsilon}
\newcommand{\est}{\mathrm{est}}
\newcommand{\poly}{\mathrm{poly}}
\newcommand{\Ber}{\mathrm{Ber}}
\newcommand{\Bin}{\mathrm{Bin}}

\newcommand{\kl}{\ensuremath{\mathrm{KL}}}
\newcommand{\dkl}{\ensuremath{D_\kl}}
\DeclarePairedDelimiterX{\infdivx}[2]{(}{)}{%
  #1\;\delimsize\|\;#2%
}
\newcommand{\DKL}{\dkl\infdivx*}

\newcommand{\dtv}{\ensuremath{d_\mathrm{TV}}}

\title{Tight simulation of a distribution using conditional samples}
\author{Tomer Adar\thanks{Technion - Israel Institute of Technology, Israel. Email: \href{mailto:tomer-adar@campus.technion.ac.il}{tomer-adar@campus.technion.ac.il}.}}

\newcommand\procnameZpreprocessZnewZsimulation{\textsf{Learn-distribution}\xspace}
\newcommand\procnameZpreprocessZnewZsimulationHREF{\hyperref[fig:alg:preprocess-simulation-init]{\procnameZpreprocessZnewZsimulation}\xspace}

\newcommand\procnameZqueryZpreprocessedZsimulation{\textsf{Query-learned-distribution}\xspace}
\newcommand\procnameZqueryZpreprocessedZsimulationHREF{\hyperref[fig:alg:preprocessed-simulation-query]{\procnameZqueryZpreprocessedZsimulation}\xspace}

\newcommand\procnameZsampleZpreprocessedZsimulation{\textsf{Sample-learned-distribution}\xspace}
\newcommand\procnameZsampleZpreprocessedZsimulationHREF{\hyperref[fig:alg:preprocessed-simulation-sample]{\procnameZsampleZpreprocessedZsimulation}\xspace}

\newcommand\procnameZinitializeZnewZsimulation{\textsf{Initialize-simulation}\xspace}
\newcommand\procnameZinitializeZnewZsimulationHREF{\hyperref[fig:alg:simulation-init]{\procnameZinitializeZnewZsimulation}\xspace}

\newcommand\procnameZaccessZsimulationZedge{\textsf{Access-simulation-edge}\xspace}
\newcommand\procnameZaccessZsimulationZedgeHREF{\hyperref[fig:alg:simulation-access-edge]{\procnameZaccessZsimulationZedge}\xspace}

\newcommand\procnameZestZsimulationZedge{\textsf{Estimate-edge}\xspace}
\newcommand\procnameZestZsimulationZedgeHREF{\hyperref[fig:alg:simulation-estimate-edge]{\procnameZestZsimulationZedge}\xspace}

\newcommand\procnameZsimulationZquery{\textsf{Query-simulation}\xspace}
\newcommand\procnameZsimulationZqueryHREF{\hyperref[fig:alg:simulation-query]{\procnameZsimulationZquery}\xspace}

\newcommand\procnameZsimulationZsample{\textsf{Sample-simulation}\xspace}
\newcommand\procnameZsimulationZsampleHREF{\hyperref[fig:alg:simulation-sample]{\procnameZsimulationZsample}\xspace}

\newcommand\paperZcoefZr{{8}}
\newcommand\paperZcoefZrZsquared{{64}}

\newcommand\paperZanticoefZqZbyZmn{{\ensuremath{1000}}}

\newcommand\paperZanticoefZmZbyZrecpZdeltasqr{{\ensuremath{400}}}

\newcommand\paperZanticoefZqZbyZnZoverZdeltasqr{{\ensuremath{4 \cdot 10^5}}}

\newcommand\paperZanticoefZfinalqZbyZnsqrZoverZepssqr{{\ensuremath{2.8 \cdot 10^{25}}}}

\newcommand\paperZGmuZerr{{\ensuremath{e^{-20}}}}
\newcommand\paperZGmuZprZestZerr{{\ensuremath{e^{-20}}}}

\begin{document}

\begin{titlepage}
    \maketitle
    \thispagestyle{empty}
    \pagestyle{empty}
    
    \begin{abstract}
        We present an algorithm for simulating a distribution using prefix conditional samples (Adar, Fischer and Levi, 2024), as well as ``prefix-compatible'' conditional models such as the interval model (Cannone, Ron and Servedio, 2015) and the subcube model (CRS15, Bhattacharyya and Chakraborty, 2018). The sample complexity is $O(\log^2 N / \eps^2)$ prefix conditional samples per query, which improves on the previously known $\tilde{O}(\log^3 N / \eps^2)$ (Kumar, Meel and Pote, 2025). Moreover, our simulating distribution is $O(\eps^2)$-close to the input distribution with respect to the Kullback-Leibler divergence, which is stricter than the usual guarantee of being $O(\eps)$-close with respect to the total-variation distance.
        
        We show that our algorithm is tight with respect to the highly-related task of estimation: every algorithm that is able to estimate the mass of individual elements within $(1 \pm \eps)$-multiplicative error must make $\Omega(\log^2 N / \eps^2)$ prefix conditional samples per element.
    \end{abstract}
    \newpage
    \tableofcontents
\end{titlepage}

\section{Introduction}
\label{sec:intro}

Distribution testing, 
formally initiated by \cite{bfrsw2000,bffkrw2001}, is a branch of property testing where the goal is distinguishing with high probability between distributions that have some property (for example, uniformity) and distributions that are far from any distribution with this property (for example, extremely biased distributions). ``Farness'' (or closeness) of distributions is usually determined by comparing the total-variation distance to a threshold parameter, usually denoted by $\eps$.

The sampling model of distribution testing allows the algorithm to obtain a sequence of independent samples from its input distribution. One core result of distribution testing in the sampling model is uniformity testing over a domain of size $N$, at the cost of $\Theta(\sqrt{N}/\eps^2)$ \cite{pan08} samples. Other elementary properties, such as identity to an explicit distribution \cite{bffkrw2001} and equivalence to an unknown distribution \cite{cdvv14}, require polynomial number of samples as well (but still sublinear).

The sampling model of distribution testing is ineffective for high-dimensional domains, such as the set of binary strings of length $n$, since it requires a polynomial number of samples in the domain size even for simple core properties, which is exponential in the dimension $n$. For obtaining sub-polynomial testing algorithms, we strengthen the model.

The conditional sampling model, independently introduced by \cite{crs15} and \cite{cfgm16}, allows the algorithm to choose a set $A$ and draw a sample from the input distribution when conditioned on belonging to $A$. Various works focus on restricted forms of the conditional sampling model, where only sets of a restricted form are eligible for conditioning. We recall a few of them.


\paragraph{Subcube conditions} The subcube conditional model defined by \cite{bc18} (see further works: \cite{ccklw21, chen2024uniformity, afl24subcube, kmp25}) is defined for distributions over product sets: the domain is $\Sigma^n$ for some alphabet $\Sigma$, and the sets eligible for condition are the cartesian product of their projections on individual characters. More formally, a subcube condition is a set of the form $A = \prod_{i=1}^n A_i$ for $A_1,\ldots,A_n \subseteq \Sigma$. Conceptually, the subcube conditional model resembles the ``\textsf{SELECT} \ldots \textsf{WHERE}''-form of SQL queries.

\paragraph{Prefix conditions} In this paper, we focus on prefix conditions \cite{afl24subcube} (implicitly in \cite{kmp25, bc18, cfgm16}). Prefix conditions are a restricted form of subcube conditions in which there exists some $k$ for which $A_1,\ldots,A_{k-1}$ are all singletons and $A_k,\ldots,A_n = \Sigma$. Note that we can identify every condition set $\{a_1\} \times \cdots \times \{a_{i-1}\} \times \Omega^{n-\abs{w}}$ with the prefix $w = a_1\ldots a_{i-1}$ (of length $i-1$) itself. 

Additional works define various condition models such as \emph{pair conditions} \cite{crs15} (the condition set must have two elements), \emph{interval conditions} \cite{crs15} ($A = \{a,\ldots,b\} \subseteq \{1,\ldots,N\}$) and \emph{coordinate conditions} \cite{bcsv24} (we fix all symbols of a string but one). For binary string domains ($\Omega=\{0,1\}^n$), prefix conditions can be seen as prefix-interval condition sets (through the encoding $x \in \{0,1\}^n \to 1 + \sum_{i=1}^n 2^{n-i} x_i$ \cite{afl24subcube}) and coordinate conditions can be seen as subcube-pair condition sets.

\subsection{Estimation and simulation of a distribution}
In the string model of property testing (for example 
\cite{rs92, rs96}), the algorithm can choose an index to query and reveal the bit in that index. In other words, it obtains certain (undoubted) information about its input. In contrast, in all sampling models, the access to the input is indirect: the algorithm only obtains ``behavioral'' information about its input.

There are two concepts of ``bridging'' between sampling and querying. The first is \emph{estimation}: given an element $x \in \Omega$, the algorithm estimates $\mu(x)$ (with high probability) within a reasonable error range using samples. The second is \emph{simulation}: the algorithm constructs an explicit distribution $\tilde{\mu}$, which is hopefully close to the input distribution $\mu$, and then queries it instead of the input distribution. Since approximating an entire distribution is too expensive, we allow a lazy construction of $\tilde{\mu}$, starting with ``$\tilde{\mu}$ can be any distribution over $\Omega$'' and then gradually adding details required to answer queries.

Estimation and simulation are highly related tasks
. While estimation demands ``bounded worst-case error (for most input elements)'', simulation of a close distribution can be seen as a ``bounded expected estimation error (over the choice of the input element)''. However, they are not fully compatible with each other:
\begin{itemize}
    \item Assume that we can estimate all probabilities with a small error. Estimations can be biased, and therefore, the estimated probability masses do not necessarily sum to $1$. Moreover, even if they do, the specific implementation of the estimator might be unable to \emph{sample} the distribution implied by the individual estimations.
    \item Assume that we can access a distribution $\tilde{\mu}$ which is close to $\mu$ with respect to some divergence measure. For common divergences such as the KL-divergence, it might be the case that most elements are poorly estimated, but the estimation errors cancel each other, which results in a relatively low divergence.
\end{itemize}

One algorithmic paradigm which is compatible with both estimation and simulation is the distribution tree\footnote{Note that a specific \emph{implementation} inspired by this paradigm is not necessarily compatible with both.}, or ``slice and dice''. First, we define a tree of sets eligible for conditioning by the model, where the root holds the whole domain $\Omega$ and every non-leaf node has two children whose disjoint union is the set it holds. Whenever an estimation $\mu(x)$ is required, we estimate the conditional mass $\mu(\textsf{CHILD}|\textsf{PARENT})$ in all nodes in the path from the root to the leaf containing the singleton $\{x\}$. To obtain the estimation, we use the chain rule: if $\{x\} = A_k \subseteq A_{k-1} \subseteq \cdots A_1 \subseteq A_0 = \Omega$, then $\Pr_\mu[x] = \prod_{i=1}^k \Pr_\mu[A_i | A_{i-1}]$. The estimation algorithm uses $\tilde{\mu}(x) = \prod_{i=1}^k \mathrm{est}(\Pr_\mu[A_i | A_{i-1}])$.

In \cite{kmp25} there is an explicit ``slice-and-dice'' algorithm for the subcube conditional sampling model. The algorithm can estimate $\mu(x)$ within $(1 \pm \eps)$-factor for every $x$, unless one or more of its marginals is too close to $0$ or to $1$. The overall mass of these $x$s is bounded by $O(\eps)$, which implies that the simulating distribution is $O(\eps)$-close to $\mu$ with respect to the total-variation distance. The cost for every query is $\tilde{O}(n^3 / \eps^2)$ subcube conditional samples.

In this work we reduce the cost of the ``slice and dice'' method to $O(n^2 / \eps^2)$ conditional samples per query and improve the accuracy of the simulating distribution to be $O(\eps^2)$-close with respect to the Kullback-Leibler divergence. Note that by Pinsker's inequality (see Lemma \ref{lemma:cheatsheet:pinsker}) it is never worse than $O(\eps)$-total-variation-close.

The KL-divergence has two advantages over the total-variation distance for intermediate analysis: its propagation over multiple dimensions is tight (compare the equality in Lemma \ref{lemma:cheatsheet:chain-rule-dkl} to the corresponding upper bound for total-variation distance, for example in \cite{bc18}), and it provides a refined divergence bound of indicator estimations (see Lemma \ref{lemma:good-expected-dkl-bernoulli} vs.\ the well-known $O(1/\sqrt{m})$-bound for total-variation distance), allowing for propagating over a higher dimension subject to the constraint $\kl = O(\poly(\eps))$. We provide our algorithm in the prefix conditional model, but it can be implemented in every conditional model that resembles the ``slice and dice'' behavior (see Section \ref{sec:prefix-reduction} for more details).

Since explicitly learning a whole distribution is inevitably expensive (as we would have to store information about every element in the domain), we prefer to use local learning. Instead of storing a whole distribution, we provide an oracle access to the learned object, and store only the information needed to describe the queried parts. When processing a query, we compute only the (missing) parts needed to respond.

\paragraph{Notations} We refer to a \emph{learning algorithm} whose input is an oracle access to a distribution $\mu$, also called the \emph{simulated distribution}, and its output is an oracle access to an output distribution $\tilde{\mu}$, which is called the \emph{simulating distribution} (of $\mu$).

\subsection{Usage example}
The distance estimation algorithms of \cite{kmp25,afl26} consist of two modules: first, an estimation procedure $x \to \est(\mu(x))$, which we call a \emph{mass query}, and second, an unconditional sampling algorithm that obtains the approximated mass of its samples and uses them to estimate the distance between the input distributions. While the first module requires a different logic for implementation in various models, the second module is fixed and stated in the following proposition.

\begin{proposition}[\cite{afl26}, implicit in \cite{kmp25}] \label{prop:est-dtv-eps^2}
    Consider the following oracle for accessing a distribution $\mu$: the oracle draws $x \sim \mu$, and returns a pair of the form $(x, (1 \pm O(\eps))\mu(x))$. There exists an algorithm whose input is such an oracle access to two distributions $\mu$ and $\tau$ that makes $O(1/\eps^2)$ oracle calls to estimate $\dtv(\mu,\tau)$ within a $\pm \eps$-additive error with probability at least $2/3$.
\end{proposition}

\subsection{Our results}

The main result of this paper is the following theorem.

\begin{theorem}[Informal statement of Theorem \ref{th:DKLzation}]
    For every sufficiently small $\eps > 0$ and sufficiently large $n \ge 1$, there exists an algorithm for simulating an input distribution over $\{0,1\}^n$ using $O(n^2/\eps^2)$ prefix conditional samples per mass query and $O(n^2/\eps^2)$ prefix conditional samples for sampling the simulating distribution. Moreover, the expected KL-divergence of the simulating distribution from the input distribution is $O(\eps^2)$.
\end{theorem}

Additionally, we show that any estimation algorithm must make $\Omega(n^2/\eps^2)$ prefix conditional samples per query. The exact meaning of an ``estimation algorithm'' is detailed in the formal statement of Theorem \ref{th:lower-bound-prefix-new}.

\begin{theorem}[Informal statement of Theorem \ref{th:lower-bound-prefix-new}]
    Every estimation algorithm must make at least $\Omega(n^2/\eps^2)$ prefix conditional samples per query.
\end{theorem}

In Section \ref{sec:overview}, we provide a technical overview and the proof roadmap.


\section{Preliminaries}
\label{sec:prelims}

In this section we formally define the notation scheme of this paper, as well as a few definitions more specific for this work. See Appendix \ref{apx:common-and-facts} for formal definition of widely used distribution notations (such as Bernoulli) and for formal statements of well known facts (such as asymptotic bounds for $\ln x$ and Pinsker's inequality).


\subsection{Basic notations}

\begin{definition}[Conditional distribution]
    Let $\mu$ be a distribution over $\Omega$ and let $A \subseteq \Omega$ be a set. The \emph{conditional distribution} $\mu|^A$ is defined as the distribution where $\Pr[x] = \frac{\Pr_\mu[x]}{\Pr_\mu[A]}$ for every $x \in A$ and $\Pr[x] = 0$ for every $x \notin A$. If $\Pr_\mu[A] = 0$, then the conditional distribution $\mu|^A$ is undefined.
\end{definition}

\begin{definition}[Projected distribution, restricted distribution]
    Let $\mu$ be a distribution over $\Omega = \prod_{i=1}^k \Omega_i$. For a non-empty set $I \subseteq \{1,\ldots,k\}$, the \emph{projected distribution} $\mu_I$ is defined as the distribution over $\prod_{i \in I} \Omega_i$ for which $\Pr[x] = \Pr_{z \sim \mu}[\forall i \in I : z_i = x_i]$ for every $x \in \prod_{i\in I} \Omega_i$. In other words, it is the distribution obtained by drawing a sample from $\mu$ and \emph{restricting} it to the entries in $I$'s indexes.
\end{definition}

The overloaded notation $\mu|^A_I$ is interpreted as $(\mu|^A)_I$. In other words, we first sample according to the condition $A$, and then restrict the result to $I$'s entries.

\begin{definition}[Product of distributions]
    Let $\mu_1,\ldots,\mu_k$ be distributions over $\Omega_1,\ldots,\Omega_k$, respectively. The \emph{product distribution} $\mu = \prod_{i=1}^k \mu_i$ is the distribution over $\Omega = \prod_{i=1}^k \Omega_i$ where, for every $(a_1,\ldots,a_k) \in \Omega$, $\mu((a_1,\ldots,a_k)) = \prod_{i=1}^n \mu_i(a_i)$.
\end{definition}

\begin{definition}[Total-variation distance]
    Let $\mu$ and $\tau$ be two distributions over a domain $\Omega$. Their \emph{total-variation divergence} is $\dtv(\mu,\tau) = \frac{1}{2} \sum_{x \in \Omega} \abs{\mu(x) - \tau(x)}$.
\end{definition}

\begin{definition}[Kullback-Leibler divergence]
    Let $\mu$ and $\tau$ be two distributions over a domain $\Omega$. The \emph{KL-divergence} of $\mu$ from $\tau$ is $\DKL{\mu}{\tau} = \E_{x \sim \mu}[\log_2 (\mu(x) / \tau(x))]$.
\end{definition}

For convenience, we abuse the notation to define KL-divergence of probabilities:
\begin{definition}[KL-divergence of probabilities (overloaded notation)]
    For $p,q \in [0,1]$, we use the notation $\dkl(p,q) = \DKL{\Ber(p)}{\Ber(q)} = p \log \frac{p}{q} + (1-p) \log \frac{1-p}{1-q}$.
\end{definition}

\subsection{Distribution testing models}

A testing algorithm must access its input using oracles. Below we define two oracles for distribution testing.

\begin{definition}[The sampling oracle]
    Given an unknown input distribution $\mu$ over a set $\Omega$, a call to the oracle results in a sample $x \in \Omega$ that distributes according to $\mu$, independently of past calls to the oracle.
\end{definition}

\begin{definition}[The conditional sampling oracle]
    Given an unknown input distribution $\mu$ over a set $\Omega$, the input of the oracle is a non-empty explicit subset $A \subseteq \Omega$ and the output is a sample $x \in \Omega$ that distributes like $\mu|^A$. If the mass of $A$ in $\mu$ is zero, then the behavior of the oracle is undefined (it is only guaranteed that $x \in A$ with probability 1).
\end{definition}

\begin{definition}[Prefix condition sets]
    Let $\Omega = \Sigma^n$ be a domain. A set $A$ is a \emph{prefix condition} if there exists a string $w \in \bigcup_{i=0}^n \Sigma^i$ for which $A = \prod_{i=1}^{\abs{w}} \{w_i\} \times \Sigma^{n-\abs{w}}$.
\end{definition}

\subsection{Algorithmic notions}

In this subsection we formally define the notions of estimation and simulation of distributions.

Estimation algorithm must be able to correctly estimate the mass of \emph{most} elements within a given \emph{multiplicative factor}.
\begin{definition}[$(\eps,c)$-estimation algorithm]
    Consider an algorithm whose input is a pair $(\mu,x)$, where $\mu$ is a distribution and $x$ is an element in the domain of $\mu$. For given parameters $0 < \eps < 1$ and $0 < c < 1$, such an algorithm is an \emph{$(\eps,c)$-estimation algorithm} if for every input distribution $\mu$ there exists a set $G_\mu$ for which $\mu(G_\mu) \ge 1 - c$ such that, for every $x \in G_\mu$, with probability at least $2/3$, the output of the algorithm for the input $(\mu,x)$ is a number in the range $(1 \pm \eps) \mu(x)$.
\end{definition}

A simulation algorithm must be defined through an interaction interface rather than a one-time call. It consists of a one-time initialization and a sequence of allowed interactions, ``sample'' and ``query''. After each interaction, the internal state of the simulation can mutate (in particular, the simulation can track its interaction history).
\begin{definition}[Distribution simulation interface] \label{def:distribution-simulation-interface}
    For a domain $\Omega$, the interface of a \emph{distribution simulation} consists of three procedures:
    \begin{itemize}
        \item Initialization: the input is a distribution $\mu$ and any additional parameters.
        \item Query: the input is an element $x \in \Omega$, and the output is a number $p \in [0,1]$.
        \item Sample: the output is an element $x \in \Omega$ and a number $p \in [0,1]$.
    \end{itemize}    
\end{definition}

\begin{definition}[Interaction notation]
    We formalize the $q$-round interaction of the caller with the simulation algorithm (Definition \ref{def:distribution-simulation-interface}) as:
    \begin{itemize}
        \item A \emph{request} sequence $(\mathrm{init}(\mu); X_1,\ldots,X_q)$, where $X_i \in \Omega \cup \{\bot\}$ represents a query request (if $X_i \in \Omega$) or a sample request ($X_i = \bot$).
        \item A \emph{response} sequence $((\hat{X}_i,\hat{P}_i))_{1\le i \le q}$, where $\hat{X}_i$ represents the sampled element (for a sample request) or the queried element $X_i$ (for a query request), and $\hat{P}_i$ represents the $p$-output.
    \end{itemize}
\end{definition}

The simulation interface lacks of constraints ensuring that the output ``looks like a simulation'', which we have to formalize separately. Ideally, an implementation of the simulation interface secretly chooses a distribution $\tilde{\mu}$ (which is hopefully a good representation of $\mu$) at the initialization phase, and then responds to sample and query interactions by sampling and querying the chosen $\tilde{\mu}$. 



\begin{definition}[The secret-distribution constraint]
    An algorithmic implementation $\mathcal A$ for the distribution simulation interface follows the \emph{secret-distribution constraint} if for every input distribution $\mu$ there exists a distribution $\mathcal A_\mu$ over distributions such that the interaction can be seen as drawing a secret distribution $\tilde{\mu} \sim \mathcal A_\mu$ in the initialization phase and then responding to sample and query requests through $\tilde{\mu}$. More precisely, given the request sequence $(\mathrm{init}(\mu); X_1,\ldots,X_q)$, the probability to obtain the response $((\hat{X}_1,\hat{P}_1),\ldots,(\hat{X}_q,\hat{P}_q))$ is exactly:
    \[
        \sum_{\tilde{\mu} : \bigwedge_{i=1}^q (\tilde{\mu}(\hat{X}_i) = \hat{P}_i)} \left(\Pr_{\mathcal A_\mu}[\tilde{\mu}] \times \prod_{i=1}^q \begin{cases}
        1 & X_i = \hat{X}_i \text{ (query)} \\
            \tilde{\mu}(\hat{X}_i) & X_i = \bot \text{ (sample)} \\
            0 & X_i \ne \hat{X}_i, \bot \text{ (violation)} \\
        \end{cases} \right)
    \]
\end{definition}

\begin{definition}[$(D,\delta)$-simulation of a distribution] \label{def:D-delta-simulation-of-distribution}
    For a given divergence measure $D$ and a parameter $\delta$, an algorithmic implementation $\mathcal A$ of the distribution simulation interface is a \emph{$(D,\delta)$-simulation} if:
    \begin{itemize}
        \item Simulation behavior: $\mathcal A$ follows the secret-distribution constraint.
        \item Closeness: for every $\mu$, regarding the distribution $\mathcal A_\mu$ (over distributions) required by the secret-distribution constraint, $\E_{\tilde{\mu} \sim \mathcal A_\mu}[D(\tilde{\mu} | \mu)] \le \delta$.
    \end{itemize}
\end{definition}


\section{Simulating conditional samples using prefix conditions}
\label{sec:prefix-reduction}

In this section we provide a conceptual reduction from conditional models compatible with the distribution-tree paradigm, which we formally define below, to the prefix conditional model.

First, we define a \emph{breakdown tree}, which is the underlying structure of a distribution tree.

\begin{definition}[A breakdown tree] \label{def:breakdown-tree}
    Let $\Omega$ be a domain. A \emph{breakdown tree} is the following combinatorial tree structure:
    \begin{itemize}
        \item All leaves have the same depth.
        \item Every node holds a subset of the domain.
        \item A leaf can hold either the empty set or a singleton.
        \item The root holds the domain set $\Omega$.
        \item Every non-leaf node has two children, arbitrarily named ``left'' and ``right''.
        \item Every non-leaf node holds the disjoint union of the subsets held by its children.
    \end{itemize}
\end{definition}

\begin{observation}
    Let $\mathcal T$ be a breakdown tree over $\Omega$. For every $x \in \Omega$, there exists a single leaf that holds the singleton $\{x\}$.
\end{observation}

A conditional sampling model is \emph{compatible with the distribution-tree paradigm} if for every set $\Omega$ which is valid as a domain of distributions in this model, there exists a breakdown tree in which every subset associated with a node is a valid condition set in the model. It is immediate to see that the subcube conditional model, the interval conditional model and the prefix conditional model are compatible with the distribution-tree paradigm, whereas the pair conditional model\footnote{The pair conditional model is defined in \cite{crs15}. Beside the whole domain, the valid condition sets are all sets with exactly two elements.} is not.

We extend the structure of a breakdown tree to define distribution trees.
\begin{definition}[A distribution tree] \label{def:distribution-tree}
    Let $\Omega$ be a domain. A \emph{distribution tree} consists of a breakdown tree and additionally:
    \begin{itemize}
        \item Every edge holds a probability.
        \item For every non-leaf node, the sum of its (two) outgoing edges is $1$.
        \item The probability of an edge from a node holding a non-empty set to a child holding the empty set is $0$.
        \item The probability of an edge from a node holding the empty set to its left child is $1$.
    \end{itemize}
\end{definition}

\begin{definition}[Representation] \label{def:distribution-tree-representation}
    A distribution tree $\mathcal T$ of height $\ell$ over $\Omega$ \emph{represents} a distribution $\mu$ over $\Omega$ if, for every $x \in \Omega$, considering the tree path $\Omega = A_0 \supseteq A_1 \supseteq \cdots \supseteq A_{\ell-1} \supseteq A_\ell = \{x\}$ and the weights $p_1,\ldots,p_\ell$ associated with the edges $A_0 \to A_1, \ldots, A_{\ell-1} \to A_\ell$, it holds that $\mu(x) = \prod_{i=1}^\ell p_i$.
\end{definition}

\begin{observation}
    Let $\mathcal T$ be a breakdown tree over a domain $\Omega$. There exists a bijection between distribution trees with the same structure of $\mathcal T$ and distributions over $\Omega$ such that every tree represents the distribution it is mapped to. Moreover, the mapped tree can be computed directly by associating every edge $A \to B$ ($A \supseteq B$) with the probability $\Pr_{z\sim \mu}[z \in B | z \in A]$.
\end{observation}

Observe that, given a breakdown tree of edge-height $\ell$ for a conditional model over a domain $\Omega$, every condition associated with a node can be seen as a prefix condition in $\{0,1\}^\ell$, where the encoding of elements $x \in \Omega$ into elements in $\{0,1\}^\ell$ corresponds to the sequence of left and right edges in the path from the root to $x$.

Therefore, every prefix conditional model algorithm can be directly executed in every conditional model that has a breakdown tree, where the $n$ parameter (for prefix algorithms over $\{0,1\}^n$) is the edge-height of the breakdown tree.

\section{Overview}
\label{sec:overview}

\subsection{The prefix conditional model}

We re-define the notions of a breakdown tree and a distribution tree to specifically match the behavior of prefix conditions.

\begin{definition}[Prefix domain]
    For $n \ge 1$, let $\{0,1\}^{<n} = \bigcup_{i=0}^{n-1} \{0,1\}^i$ be the set of all true prefixes of strings in $\{0,1\}^n$.
\end{definition}

Observe that there is a natural bijection between strings in $\{0,1\}^{<n}$ and non-singleton prefix conditions in $\{0,1\}^n$: a string $w \in \{0,1\}^{<n}$ corresponds to the prefix condition set $\{w\} \times \{0,1\}^{n-\abs{w}}$. Moreover, the prefix conditional model has exactly one breakdown tree: the tree in which the non-leaf nodes are the sets corresponding to binary strings of length strictly less than $n$, the leaf nodes are the elements of $\{0,1\}^n$, and the children of a non-leaf node $w$ correspond to the prefixes $w0$ and $w1$.

\begin{definition}[Marginal tree] \label{def:marginal-tree}
    Let $\Omega = \{0,1\}^n$. A marginal tree consists of the (unique) breakdown tree and a function $\{0,1\}^{<n} \to [0,1]$ representing the probabilities associated with the one-edges (the edges of the form $w \to w1$).
\end{definition}

\begin{observation}[Representation] \label{obs:marginal-tree-representation}
    Given the marginal function $f : \{0,1\}^{<n} \to [0,1]$, the marginal tree defines a distribution in which, for every $x \in \{0,1\}^n$, \[\Pr[x] = \prod_{i=1}^n (x_i f(x_{1,\ldots,i-1}) + (1 - x_i) (1 - f(x_{1,\ldots,i-1})))
    = \prod_{i=1}^n \begin{cases}
        f(x_{1,\ldots,i-1}) & x_i = 1 \\
        1 - f(x_{1,\ldots,i-1}) & x_i = 0 \\
    \end{cases}
    \]
\end{observation}

We use the specific structure of the only breakdown tree in the prefix conditional model to introduce a notation for marginal probabilities.

\begin{definition}[Overloaded notation for marginals]
    The marginal function $f : \{0,1\}^{<n} \to [0,1]$ represents a single distribution $\mu$ over $\{0,1\}^n$. For a prefix $w$ of length $i-1$, we use the notation $\mu(x_i = 1 | x_{1,\ldots,i-1} = w)$, its short form $\mu(x_i = 1 | w)$ and its shorter form $\mu(1|w)$, to denote $f(w)$. Analogously, $\mu(x_i = 0 | x_{1,\ldots,i-1} = w)$, $\mu(x_i = 0 | w)$ and $\mu(0|w)$ denote $1-f(w)$.
\end{definition}

In both the upper bound and the lower bound analysis, we use the chain rule for distributions over strings.

To \emph{query} the mass of $x \in \{0,1\}^n$ with respect to $\mu$, we multiply its marginal probabilities:
\[  \mu(x)
    = \prod_{i=1}^n \mu(x_i | x_{1,\ldots,i-1})
    = \prod_{i=1}^n (x_i f(x_{1,\ldots,i-1}) + (1 - x_i) (1 - f(x_{1,\ldots,i-1}))) \]
To \emph{estimate} $\mu(x)$, we multiply the estimations of its marginals.
\[  \est(\mu(x))
    = \prod_{i=1}^n \est(\mu(x_i | x_{1,\ldots,i-1}))
    = \prod_{i=1}^n (x_i \est(f(x_{1,\ldots,i-1})) +
    (1 - x_i) (1 - \est(f(x_{1,\ldots,i-1})))) \]
Note that the estimations of the marginals can be fully independent.

In parts of our analysis we use a weaker form of the prefix conditional sampling oracle.

\begin{definition}[The marginal prefix sampling oracle]
    Let $\mu$ be a distribution over $\Omega = \{0,1\}^n$. The input of a \emph{marginal prefix sampling oracle} is a prefix condition $\{w\} \times \{0,1\}^{n-\abs{w}}$ (for some $w \in \{0,1\}^{<n}$, and its output is a single bit that distributes like $\mu|^{w\times \{0,1\}^{n-\abs{w}}}_{\abs{w}+1}$.
\end{definition}
In other words, a marginal prefix oracle call draws a prefix conditional sample and then ignores all but the first free bit of the sample.

We can convert any prefix conditional sampling (including an unconditional sampling, in which the condition corresponds to the empty prefix) to a (random) sequence of marginal prefix samples whose length is the same as the number of free bits. Let $X_1,\ldots,X_n$ be the random variables for which the result sample is the concatenation $x = X_1\cdots X_n$. Let $w \in \{0,1\}^{<n}$ of length $k-1$ ($1 \le k \le n$) be the prefix to condition on. For $1 \le i \le k-1$, $X_i$ is the fixed value $w_i$. For $k \le i \le n$, once $X_1,\ldots,X_{i-1}$ are determined, we can sample $X_i$ as the first bit of a sample conditioned on the prefix $X_1 \cdots X_{i-1}$.

\subsection{The upper bound (Section \ref{sec:dklzation})}
We show that we can estimate a single marginal probability using $m$ samples such that the expected KL-divergence of the estimation from the true marginal probability is bounded by $1/m$. By the chain rule, if we estimated all $2^n - 1$ marginal probabilities independently, then the expected KL-divergence of the result distribution from the input distribution would be bounded by $n/m$. We choose $m=n/\eps^2$ to have an $O(\eps^2)$-$\dkl$-close distribution in expectation at the cost of $O(n/\eps^2)$ prefix samples per marginal estimation. Since the marginal estimations are fully independent, given an element $x$ we can estimate only its $n$ marginals at the cost of $O(n^2/\eps^2)$ prefix samples.

\subsection{The lower bound (Section \ref{sec:lower-bound-estimation}, Appendix \ref{apx:ad-hoc-task})}

Recall that the input for the algorithm is pair of (1) a distribution over $\Omega = \{0,1\}^n$ and (2) an element $x \in \Omega$, and the output should be a number in the range $(1 \pm \eps)\mu(x)$ with probability $2/3$.

We first describe the construction. We define two internal parameters, $\delta=\sqrt{2}\eps/\sqrt{n}$ and $r=\paperZcoefZr/\sqrt{n}$. For every $w \in \{0,1\}^{<n}$ we draw $s_w \in \{+1, -1\}$ uniformly and independently. We define the two distributions $\mu$ and $\mu'$ through its marginals. For every $w \in \{0,1\}^n$, we define:
\begin{itemize}
    \item $\mu(1|w) = \frac{1}{2}(1 + s_w \delta)$, so that $\Pr_\mu[x] = 2^{-n} \prod_{i=1}^n (1 + (-1)^{1-x_i} s_{x_{1,\ldots,i-1}} \delta)$.
    \item $\mu'(1|w) = \frac{1}{2}(1 - s_w r)$, so that $\Pr_{\mu'}[x] = 2^{-n} \prod_{i=1}^n (1 + (-1)^{x_i} s_{x_{1,\ldots,i-1}} r)$.
\end{itemize}

We define two distributions over inputs, based on the choice of $s_w$ variables:
\begin{itemize}
    \item $D_\mathrm{yes}$: the input provided to the algorithm is $(\mu,x)$, where $x$ is drawn uniformly.
    \item $D_\mathrm{no}$: the input provided to the algorithm is $(\mu,x)$, where $x$ is drawn according to $\mu'$.
\end{itemize}

The construction guarantees different concentration ranges for $\mu(x)$, as stated in the following lemma.
\begin{lemma}[See Lemma \ref{lemma:div-pH-pL}, \ref{lemma:dyes-concentration-ge-pH}, \ref{lemma:dno-concentration-le-pL}]
    For every sufficiently large $n \ge N_0$ and sufficiently small $0 < \eps \le \eps_0$ there exist $p_L$ and $p_H$ for which:
    \begin{itemize}
        \item $(1+\eps)p_L < (1-\eps)p_H$.
        \item If $(\mu,x)$ is drawn from $D_\mathrm{yes}$, then $\mu(x) \ge p_H$ with probability at least $1 - e^{-7.95}$.
        \item If $(\mu,x)$ is drawn from $D_\mathrm{no}$, then $\mu(x) \le p_L$ with probability at least $e^{-7.38}$.
    \end{itemize}
\end{lemma}

To prove a lower bound, we recall the meaning of a good estimation algorithm:
\begin{itemize}
    \item The set of ``good'' elements has mass $1-O(\eps)$ according to the input distribution.
    \item Given a ``good'' element $x$, the output is in the range $(1 \pm \eps)\mu(x)$ with probability at least $\frac{2}{3}$.
\end{itemize}
Note that the success probability in second requirement can be strengthen to $1-c$ (instead of $2/3$) at the cost of $O(\log (1/c))$, which is asymptotically ignorable when $c$ is a constant.

Consider an estimation algorithm $\mathcal A$ that can estimate $\mu(x)$ within a $(1 \pm \eps)$-factor with probability at least $1 - e^{-20}$, unless $x$ belongs to a small set ``bad'' elements whose mass is bounded by $e^{-20}$. In this setting,
\begin{eqnarray*}
    \dtv(\mathcal A(D_\mathrm{yes}), \mathcal A(D_\mathrm{no}))
    &\ge& \Pr\left[\mathcal A(D_\mathrm{yes}) \ge (1-\eps)p_H\right] - \Pr\left[\mathcal A(D_\mathrm{yes}) \ge (1-\eps)p_H\right] \\
    &\ge& (e^{-7.38} - e^{-7.95}) - 4e^{-20}
    \ge e^{-8.25}
\end{eqnarray*}


Therefore, if an algorithm makes too few queries to distinguish between $D_\mathrm{yes}$ and $D_\mathrm{no}$ with total-variation distance at least $e^{-8.25}$, then it cannot be a good estimation algorithm. It remains to describe the lower bound for distinguishing between $D_\mathrm{yes}$ and $D_\mathrm{no}$ using prefix queries.

%

Let $I_x$ be the set of prefixes (elements in $\{0,1\}^{<n}$) that are not prefixes of $x$. When $x$ is explicitly given as an input, $\mu(x)$ is fully determined by the random sequence $(s_{x_{1,\ldots,i-1}})_{1 \le i \le n}$, which may depend on the choice of $x$. Also, note that if $x$ is drawn uniformly, then $\mu(x)$ is independent of the assignment of all $s_w$s for $w \in I_x$. We observe that this independence is preserved even if $x$ is drawn from $\mu'$ (Lemma \ref{lemma:bit-by-bit-rest-is-independent}).

Based on this analysis, we can define a random sequence $p_i = \mu(x_{1,\ldots,i} | x_{1,\ldots,i-1})$ (for $1 \le i \le n$).
\begin{itemize}
    \item For $D_\mathrm{yes}$, $p_i = (1 + \delta)/2$ with probability $1/2$ and otherwise $p_i = (1 - \delta)/2$.
    \item For $D_\mathrm{no}$, $p_i = (1 + \delta)/2$ with probability $(1 - r)/2$ and otherwise $p_i = (1 - \delta)/2$.
\end{itemize}
Clearly, distinguishing between $D_\mathrm{yes}$ and $D_\mathrm{no}$ cannot be easier than distinguishing between the sequences $(p_1,\ldots,p_n)$ drawn from $D_\mathrm{yes}$ and from $D_\mathrm{no}$.

To show a lower bound for distinguishing between the $p_i$-sequence inputs, we use the following model that has been introduced in \cite{lv19uncertainty}: the input is a sequence of $n$ coins, and in every step, the algorithm can choose a coin and sample it. Alternatively, the input can be seen as a sequence of inaccessible probabilities $p_1,\ldots,p_n \in [0,1]$, and in every step, the algorithm can choose an index $1 \le i \le n$ and sample a single bit that distributes like $\Ber(p_i)$.

We use the following reduction for simulating the coin model using the prefix sampling model. The input of the distinguishing task corresponds to the sequence of marginals of a given random element $x$ (that is, $p_1,\ldots,p_n$). Every prefix conditional sample translates into $O(1)$ steps in the coin model in expectation (see Lemma \ref{lemma:expected-number-of-effective-marginal-samples}). 

We can use the following result regarding a closely related task to deduce an $\Omega(1/r^2 \delta^2) = \Omega(n^2/\eps^2)$ lower bound. However, we explicitly prove hardness in Appendix \ref{apx:ad-hoc-task} for completeness.
\begin{lemma}[Rephrase of Theorem 3 in \cite{lv19uncertainty}]
    \label{lemma:lv21-coin-lower-bound}
    Assume that for every $1 \le i \le n$, $p_i \ge (1 + \delta)/2$ or $p_i \le (1 - \delta)/2$. In this setting, estimating the fraction of positively-biased ($p_i \ge (1+\delta)/2$) within a $\pm r$-additive error requires $\Omega(1/r^2 \delta^2)$ samples.
\end{lemma}

\subsection{Parameter table}
Different sections use different parameter constraints. The upper bound uses $n$ and $\eps$. The lower bound for the sequence-distinguishing task uses $n$, $\delta$ and $r$, whereas the lower bound for estimation uses $n$ and $\eps$ and explicitly defines $\delta$, $r$.

\begin{table}[H]
    \centering
    \begin{tabularx}{\textwidth}{l|XXXX}
        \textbf{Section} & $n$ & $\eps$ & $\delta$ & $r$ \\
        \hline
        \textbf{\ref{sec:dklzation} Upper bound} & $n \ge 1$ & $\eps > 0$ & N/A & N/A \\
        \textbf{\ref{sec:lower-bound-estimation} Lower bound} & $n \ge 4.5 \cdot 10^{16}$& $0 < \eps \le e^{-11}$ & $\sqrt{2}\eps/\sqrt{n}$ & $\paperZcoefZr/\sqrt{n}$ \\
        \textbf{\ref{apx:ad-hoc-task} Sequence distinguishing} & by context & N/A & $0 < \delta < 1/3$ & $0 \le r < 1/12$
    \end{tabularx}
\end{table}

\section{The simulation algorithm}
\label{sec:dklzation}

In this section we describe our simulation algorithm in the prefix conditional sampling model for an input distribution $\mu$ over $\Omega = \{0,1\}^n$.

Before we introduce our algorithm, we describe the common method (which is compatible with \cite{kmp25}). For every $1 \le i \le n$ and every $w \in \{0,1\}^{i-1}$, we use $O(n^2 \log (n/\eps) / \eps^2)$ samples to estimate the marginal probability $\mu(x_i = 1 | w)$ within a $\left(1 \pm \frac{\eps}{\sqrt{n}}\right)$-multiplicative factor with probability at least $1 - O(\eps/n)$ (unless it is too close to zero or one, and in this case we estimate it as zero or one with probability $1 - O(\eps/n)$).

Along the path from the root to a given leaf there are $n$ edges, but errors of different directions (overshooting and undershooting) mostly cancel each other. Hence, with high probability, the accumulated error (when multiplying all $n$ errors) in the range is $\left(1 \pm \frac{\eps}{\sqrt{n}}\right)^{\pm \Theta(\sqrt{n})} = 1 \pm O(\eps)$. By the union bound, the probability to obtain good estimations in all edges is at least $1 - O(\eps)$. Overall, all elements have their mass estimated within $(1 \pm O(\eps))$-multiplicative error, except for a set of elements whose expected mass is $O(\eps)$ as well. This implies that the expected total-variation distance of $\mu$ and the simulating distribution is $O(\eps)$. The cost of querying the simulating distribution is $O(n^3 \log (n/\eps) / \eps^2)$, since we have to query $n$ edges at the cost of $O(n^2 \log (n/\eps) / \eps^2)$ prefix samples per edge.

Our simulation algorithm uses the chain rule for $\dkl$ (Lemma \ref{lemma:cheatsheet:chain-rule-dkl}). Similarly to \cite{bc18, afl24subcube}, repeatedly applying the chain rule results in a bitwise identity.
\begin{lemma}[Repeated chain rule for $\dkl$] \label{lemma:dkl-repeated-chain-rule}
    Let $\mu$ and $\tau$ be two distributions over $\{0,1\}^n$. In this setting, $\DKL{\mu}{\tau} = \sum_{i=1}^n \E_{w\sim \mu|_{1,\ldots,i-1}}\left[\DKL{\mu|_i^{\{w\} \times \{0,1\}^{n-i}}}{\tau|_i^{\{w\} \times \{0,1\}^{n-i}}} \right]$.
\end{lemma}

While the total-variation distance of Bernoulli distributions is linear, Lemma \ref{lemma:cheatsheet:dkl-quadratic-bound} states that the KL-divergence has a $\chi^2$-like behavior, which is sublinear in many cases.

The following lemma is the core of our algorithm. It states that the expected KL-divergence of an estimated probability mass is linear in the number of samples, rather than its square-root as in total-variation distance.

\begin{lemma} \label{lemma:good-expected-dkl-bernoulli}
    Let $p \in [0,1]$ and $m \ge 1$. In this setting, $\E_{t \sim \Bin(m,p)}\left[\dkl(\frac{t}{m}, p)\right] \le \frac{1}{m}$.
\end{lemma}
\begin{proof}
    For $p=0$ and $p=1$, $\Pr[t=mp]=1$, and hence $\E[\dkl(\frac{t}{m}, p)] = 0 \le \frac{1}{m}$. For the rest of the proof we assume that $0 < p < 1$. In particular, we can safely divide by $p$ and by $1-p$.

    Let $\hat{p} = \frac{t}{m}$ and $\delta = \frac{\hat{p}}{p} - 1$ (so that $\hat{p} = (1 + \delta)p$). By Lemma \ref{lemma:cheatsheet:dkl-quadratic-bound}, $\dkl(\hat{p},p) \le \frac{(\hat{p}-p)^2}{p(1-p)} = \frac{\delta^2 p}{1-p}$. Therefore,
    \begin{eqnarray*}
        \E_t[\dkl(t/m,p)]
        \le \frac{p}{1-p} \cdot \E_t\left[\left(\frac{t/m}{p} - 1\right)^2\right]
        &=& \frac{p}{(mp)^2 (1-p)} \cdot \E_t\left[\left(t - \E[t]\right)^2\right] \\
        &=& \frac{1}{m^2 p (1-p)} \Var[t]
        = \frac{1}{m^2 p (1-p)} p(1-p)m
        = \frac{1}{m}
    \end{eqnarray*}
\end{proof}

We show two implementations of $(\dkl,\delta)$-simulation. The first, shown in Algorithms \ref{fig:alg:preprocess-simulation-init}, \ref{fig:alg:preprocessed-simulation-query} and \ref{fig:alg:preprocessed-simulation-sample} as a conceptual introduction, is extremely inefficient, but it is quite easy to show its correctness as a learning algorithm. The second, in Algorithms \ref{fig:alg:simulation-init}, \ref{fig:alg:simulation-query} and \ref{fig:alg:simulation-sample}, is efficient, and we prove its correctness by coupling its behavior with the behavior of the first implementation.

Both implementations use a common procedure $\procnameZestZsimulationZedgeHREF$ (Algorithm \ref{fig:alg:simulation-estimate-edge}) that estimates the probability mass of a single edge by drawing $m = \ceil{n/\delta}$ prefix conditional samples and averaging the first free bit.

\begin{algo}
    \procname{$\procnameZestZsimulationZedge (n, \mu, \delta; w, b)$}
    \label{fig:alg:simulation-estimate-edge}
    \alginput{$\mu$ is the distribution over $\Omega=\{0,1\}^n$}
    \alginput{$\delta$ is the accuracy parameter}
    \alginput{$w \to wb$ is the edge to estimate}
    \algoutput{An estimation of $\mu(b | w)$}
    \begin{code}
        \algitem Let $m \gets \ceil{n / \delta}$.
        \algitem Draw $x_1,\ldots,x_m$ independently from $\mu$, conditioned on $x_{1,\ldots,\abs{w}} = w$.
        \algitem Let $k$ be the number of samples $x_1,\ldots,x_m$ in which the $\abs{w}+1$st bit is $b$.
        \algitem Return $k/m$.
    \end{code}
\end{algo}

The first implementation has the learn-then-read concept. We exhaustively build the distribution tree of a distribution $\tilde{\mu}$ simulating the input distribution $\mu$, and then query $\tilde{\mu}$ at individual elements whenever needed.

The initialization procedure, \procnameZpreprocessZnewZsimulationHREF, estimates all right edges (of the form $w \to w1$) independently.

\begin{algo}
    \procname{$\procnameZpreprocessZnewZsimulation (n,\mu,\delta)$}
    \label{fig:alg:preprocess-simulation-init}
    \alginput{$\mu$ is the distribution over $\Omega = \{0,1\}^n$ to simulate}
    \alginput{$\delta$ is the accuracy parameter}
    \algoutput{The simulating distribution $\tilde{\mu}$}
    \algcontract{Accuracy}{$\E\left[\DKL{\tilde{\mu}}{\mu}\right] \le \delta$}
    \begin{code}
        \algitem Initialize the marginal function $f : \{0,1\}^{<n} \to [0,1]$, with all entries initially undefined.
        \begin{For}{$w \in \{0,1\}^{<n}$ in an arbitrary order}
            \algitem Run and write $f(w) \gets \procnameZestZsimulationZedgeHREF(n,\mu,\delta; w,1)$.
        \end{For}
        \algitem Let $\tilde{\mu}$ be the distribution represented by the marginal function $f$. \algcomment (Def.\ \ref{def:marginal-tree}, Obs.\ \ref{obs:marginal-tree-representation})
        \algitem Return $\tilde{\mu}$.
    \end{code}
\end{algo}

The query procedure \procnameZqueryZpreprocessedZsimulationHREF and the sample procedure \procnameZsampleZpreprocessedZsimulationHREF just access the explicitly-written learned distribution $\tilde{\mu}$.

\begin{algo}
    \procname{$\procnameZqueryZpreprocessedZsimulation (\tilde{\mu},x)$}
    \label{fig:alg:preprocessed-simulation-query}
    \alginput{$\tilde{\mu}$ is the learned distribution returned by \procnameZpreprocessZnewZsimulation}
    \algoutput{The probability mass of $x$ according to $\tilde{\mu}$}
    \begin{code}
        \algitem Return $\tilde{\mu}(x)$.
    \end{code}
\end{algo}

\begin{algo}
    \procname{$\procnameZsampleZpreprocessedZsimulation (\tilde{\mu})$}
    \label{fig:alg:preprocessed-simulation-sample}
    \alginput{$\tilde{\mu}$ is the learned distribution returned by \procnameZpreprocessZnewZsimulation}
    \algoutput{A sample $x$ drawn from the learned distribution, and its probability mass $\tilde{\mu}(x)$}
    \begin{code}
        \algitem Draw $x$ from $\tilde{\mu}$.
        \algitem Return $(x, \tilde{\mu}(x))$.
    \end{code}
\end{algo}

\begin{lemma} \label{lemma:dklization-core}
    The triplet (\procnameZpreprocessZnewZsimulationHREF, \procnameZqueryZpreprocessedZsimulationHREF, \procnameZsampleZpreprocessedZsimulationHREF) is a $(\dkl,\delta)$-simulation of its input distribution.
\end{lemma}
\begin{proof}
    Since the querying and sampling algorithms are straightforward, we only have to show that the initialization procedure is correct. That is, $\E\left[\DKL{\tilde{\mu}}{\mu}\right] \le \delta$.
    
    Noting that $\tilde{\mu}$ is defined using a tree, we abuse the notation of $\tilde{\mu}(w)$ to denote ``the probability to reach a leaf in the subtree represented by $w$''. More formally, for $w \in \{0,1\}^{<n}$, let $\tilde{\mu}(w) = \Pr_{x \sim \tilde{\mu}}\left[x \in w \times \{0,1\}^{n-\abs{w}}\right]$.
    
    For every $1 \le i \le n$ and $w \in \{0,1\}^{i-1}$, let $X_w = \DKL{\tilde{\mu}|^{w\times\{0,1\}^{n-i}}_i}{ \mu|^{w\times\{0,1\}^{n-i}}_i}$. By the repeated chain rule (Lemma \ref{lemma:dkl-repeated-chain-rule}), 
    \begin{eqnarray*}
        \DKL{\tilde{\mu}}{\mu}
        &=& \sum_{i=1}^n \E_{w \sim \tilde\mu|_{1,\ldots,i-1}}[X_w] \\
        &=& \sum_{i=1}^n \sum_{w \in \{0,1\}^{i-1}} \tilde\mu(w) X_w
    \end{eqnarray*}
    
    Observe that, for every $w \in \{0,1\}^{<n}$, $\tilde{\mu}(w)$ and $X_w$ are independent, since they involve disjoint sets of edges, and therefore:
    \begin{eqnarray*}
        \E_{\text{choice of $\tilde{\mu}$}} \left[\DKL{\tilde{\mu}}{\mu}\right]
        &=& \sum_{i=1}^n \sum_{w\in \{0,1\}^{i-1}} \E_{\text{choice of $\tilde{\mu}$}}[\tilde{\mu}(w) X_w] \\
        \text{[Since $\tilde{\mu}(w)$, $X_w$ are independent]} &=& \sum_{i=1}^n \sum_{w\in \{0,1\}^{i-1}} \E_{\text{choice of $\tilde{\mu}$}}[\tilde{\mu}(w)] \E_{\text{choice of $\tilde{\mu}$}}[X_w] \\
        \text{[Lemma \ref{lemma:good-expected-dkl-bernoulli}]} &\le& \sum_{i=1}^n \sum_{w\in \{0,1\}^{i-1}} \E_{\text{choice of $\tilde{\mu}$}}[\tilde{\mu}(w)] \cdot \frac{1}{m} \\
        \text{[Linearity of expectation]} &=& \frac{1}{m} \sum_{i=1}^n \E_{\text{choice of $\tilde{\mu}$}} \left[ \sum_{w\in \{0,1\}^{i-1}} \tilde{\mu}(w) \right]
        = \frac{1}{m} \cdot \sum_{i=1}^n 1
        = \frac{n}{m}
    \end{eqnarray*}

    Since Algorithm \ref{fig:alg:preprocess-simulation-init} uses $m = \ceil{n/\delta}$, $\E\left[\DKL{\tilde{\mu}}{\mu}\right] \le \delta$.
\end{proof}

The second implementation has the lazy-construction concept. Since the estimations of Algorithm \ref{fig:alg:preprocess-simulation-init} are independent, it suffices to evaluate an edge when it (or its sibling edge) lies on a path of a query for the first time.

The initialization procedure \procnameZinitializeZnewZsimulationHREF initializes an empty memorization object.

\begin{algo}
    \procname{$\procnameZinitializeZnewZsimulation (n,\mu,\delta)$}
    \label{fig:alg:simulation-init}
    \alginput{$\mu$ is the distribution over $\Omega = \{0,1\}^n$ to simulate}
    \alginput{$\delta$ is the accuracy parameter}
    \begin{code}
        \algitem Return $(n,\mu,\delta,\mathit{hist})$, where $\mathit{hist}$ is an empty list.
    \end{code}
\end{algo}

The sample and query procedures internally call the \procnameZaccessZsimulationZedgeHREF procedure for an edge-by-edge execution.

The \procnameZaccessZsimulationZedgeHREF procedure allows a memorized estimation of an edge. If the requested edge has never been estimated, then the procedure estimates it using \procnameZestZsimulationZedgeHREF and then memorizes the result for both that edge and its sibling edge.

\begin{algo}
    \procname{$\procnameZaccessZsimulationZedge (\mathit{sim}, w, b)$}
    \label{fig:alg:simulation-access-edge}
    \alginput{$\mathit{sim}$: the simulation object obtained by \procnameZinitializeZnewZsimulationHREF}
    \alginput{$w \to wb$ is edge to access}
    \algcontract{Side effects}{The $\mathit{hist}$ component of $\mathit{sim}$ may change}
    \begin{code}
        \algitem Let $n,\mu,\delta,\mathit{hist}$ be the components of $\mathit{sim}$ as a $4$-tuple.
        \begin{If}{$\mathit{hist}$ contains $((w,b),\mathit{ans})$}
            \algitem Let $f \gets \mathit{ans}$.
        \end{If}
        \begin{Else}
            \algitem Let $f \gets \procnameZestZsimulationZedgeHREF(n,\mu,\delta;w,b)$.
            \algitem Add $((w,b), f)$, $((w,1-b), 1 - f)$ to $\mathit{hist}$.
        \end{Else}
        \algitem Return $f$.
    \end{code}
\end{algo}

The query and sample procedures execute an edge-by-edge logic. The edges are estimated through \procnameZaccessZsimulationZedgeHREF.

\begin{algo}
    \procname{$\procnameZsimulationZquery (\mathit{sim},x)$}
    \label{fig:alg:simulation-query}
    \alginput{$\mathit{sim}$: the simulation object obtained by \procnameZinitializeZnewZsimulationHREF}
    \algoutput{The probability mass of $x$ according to the secret distribution}
    \algcontract{Side effects}{The $\mathit{hist}$ component of $\mathit{sim}$ may change}
    \begin{code}
        \algitem Let $n,\mu,\delta,\mathit{hist}$ be the components of $\mathit{sim}$ as a $4$-tuple.
        \algitem Set $p_x \gets 1$.
        \begin{For}{$i$ from $1$ to $n$}
            \algitem Let $f_{x,i} \gets \procnameZaccessZsimulationZedgeHREF(\mathit{sim},x_{1,\ldots,i-1}, x_i)$.
            \algitem Set $p_x \gets p_x \cdot f_{x,i}$.
        \end{For}
        \algitem Return $p_x$.
    \end{code}
\end{algo}

\begin{algo}
    \procname{$\procnameZsimulationZsample(\mathit{sim})$}
    \label{fig:alg:simulation-sample}
    \alginput{$\mathit{sim}$: the simulation object obtained by \procnameZinitializeZnewZsimulationHREF}
    \algoutput{A sample $x$ drawn from the secret distribution, and its probability mass according to that distribution}
    \algcontract{Side effects}{The $\mathit{hist}$ component of $\mathit{sim}$ may change}
    \begin{code}
        \algitem Let $n,\mu,\delta,\mathit{hist}$ be the components of $\mathit{sim}$ as a $4$-tuple.
        \algitem Set $p \gets 1$.
        \algitem Set $x \gets \text{the empty word}$.
        \begin{For}{$i$ from $1$ to $n$}
            \algitem Let $f_i \gets \procnameZaccessZsimulationZedgeHREF(\mathit{sim},x,1)$ .
            \algitem Draw $b_i \in \{0,1\}$ according to $\Ber(f_i)$.
            \algitem Set $x \gets xb_i$. \algcomment Append $b$ to the end of $x$.
            \begin{If}{$b_i = 1$}
                \algitem Set $p \gets p \cdot f_i$.
            \end{If}
            \begin{Else}
                \algitem Set $p \gets p \cdot (1 - f_i)$.
            \end{Else}
        \end{For}
        \algitem Return $(x, p)$.
    \end{code}
\end{algo}

\begin{theorem} \label{th:DKLzation}
    The triplet (\procnameZinitializeZnewZsimulationHREF, \procnameZsimulationZqueryHREF, \procnameZsimulationZsampleHREF) has the exact behavior (with respect to the user) as the triplet (\procnameZpreprocessZnewZsimulationHREF, \procnameZqueryZpreprocessedZsimulationHREF, \procnameZsampleZpreprocessedZsimulationHREF). In particular, it is a $(\dkl,\delta)$-simulation of its input distribution.
\end{theorem}
\begin{proof}
    The interaction of \procnameZsimulationZqueryHREF and \procnameZsimulationZsampleHREF with the simulation object is only done through \procnameZaccessZsimulationZedgeHREF calls. Since in Algorithm \ref{fig:alg:preprocess-simulation-init} the edges are independent, the specific order of assigning their weight is irrelevant. Each call to \procnameZaccessZsimulationZedgeHREF assigns a weight to two sibling edges, independently of other edges. Therefore, the sequential access to the constructed distribution tree through \procnameZaccessZsimulationZedgeHREF is identically distributed as evaluating all edges in advance as in \procnameZpreprocessZnewZsimulationHREF followed by reading the preprocessed data.
\end{proof}

We improve on \cite{kmp25} by applying Theorem \ref{th:DKLzation} to Proposition \ref{prop:est-dtv-eps^2}.

\begin{corollary}
    Let $\mu$ and $\tau$ be two distributions over $\Omega = \{0,1\}^n$. Estimating the total-variation distance of $\mu$ and $\tau$ up to a $\pm \eps$-additive error with probability at least $2/3$ can be done using $O(n^2 / \eps^4)$ prefix conditional samples.
\end{corollary}
\begin{proof}
    We use Theorem \ref{th:DKLzation} to simulate $\mu$ and $\tau$ using random distributions $\tilde{\mu}$ and $\tilde{\tau}$ such that $\E[\DKL{\tilde{\mu}}{\mu}] < \frac{1}{36}\eps^2$ and $\E[\DKL{\tilde{\tau}}{\tau}] < \frac{1}{36}\eps^2$. After we initialize the simulation at no sample cost, we run a $O(1/\eps^2)$-test of Proposition \ref{prop:est-dtv-eps^2} to estimate the total-variation distance between $\tilde{\mu}$ and $\tilde{\tau}$ within $\pm \frac{1}{3}\eps$-additive error. Such a test appears in \cite{kmp25} (inlined in a bigger subroutine) and in \cite{afl26} (explicitly).

    The complexity of the test is $O(1/\eps^2)$ samples of the simulating distributions, each costing $O(n^2/\eps^2)$ prefix conditional samples of the input distributions.
    
    By Markov's inequality, with probability at least $(7/8)^2$, both $\DKL{\tilde\mu}{\mu}$ and $\DKL{\tilde\tau}{\tau}$ are at most $\frac{2}{9}\eps^2$, and in this case,
    \[  \abs{\dtv(\mu,\tau) - \dtv(\tilde{\mu},\tilde{\tau})}
        \le \dtv(\tilde{\mu},\mu) + \dtv(\tilde{\tau},\tau)
        \le 2 \cdot \frac{1}{3}\eps
    \]
    
    Hence, estimating $\dtv(\tilde{\mu},\tilde{\tau})$ within $\pm \frac{1}{3}\eps$-additive error suffices to deduce a $\pm \eps$-additive error estimation of $\dtv(\mu,\tau)$. If we use $O(1)$ trials to amplify the success probability of estimating $\dtv(\tilde\mu, \tilde\tau)$ to $8/9$, then the success probability is at least $(7/8)^2 \cdot (8/9) > 2/3$.
\end{proof}
In particular, estimating the total-variation distance between two distributions over binary strings using subcube conditional sampling requires $O(n^2 / \eps^4)$ samples, a polynomial gain over the previously best known distance estimation algorithm for subcube conditions (by \cite{kmp25}), whose complexity is $\tilde{O}(n^3 / \eps^5)$ subcube conditional samples.
\section{Lower bound}
\label{sec:lower-bound-estimation}

We show that estimation of the probability mass of individual elements using prefix conditional samples cannot be more efficient than simulating the input distribution as a whole (through Algorithms \ref{fig:alg:simulation-init}, \ref{fig:alg:simulation-query}, \ref{fig:alg:simulation-sample} and Theorem \ref{th:DKLzation}).

Recall that we can fully define a distribution by assigning weights to the edges of the breakdown tree implied by the prefix condition sets.

\begin{definition}[The counter-estimation distributions over inputs]
    Let $\Omega = \{0,1\}^n$ be the domain and $\delta=\sqrt{2}\eps/\sqrt{n}$, $r=\paperZcoefZr/\sqrt{n}$ be internal parameters. For every $w \in \{0,1\}^{<n}$, we draw $s_w \in \{+1, -1\}$ uniformly and independently. We define the input distribution $\mu$, and an additional distribution $\mu'$, through their marginals. For every $w \in \{0,1\}^{<n}$, we define $\mu(x_{\abs{w}+1} = 1 | w) = \frac{1 + s_w \delta}{2}$ and $\mu'(x_{\abs{w}+1} = 1 | w) = \frac{1 - s_w r}{2}$. In $D_\mathrm{yes}$, we draw $x$ uniformly and then return the pair $(\mu,x)$. In $D_\mathrm{no}$, we draw $x \sim \mu'$ and return the pair $(\mu,x)$.
\end{definition}
We insist that $\mu$ is the input distribution in both cases. We only use $\mu'$ for drawing $x$ in $D_\mathrm{no}$.

\begin{lemma} \label{lemma:bit-by-bit-rest-is-independent}
    For every $x \in \{0,1\}^n$, let $I_x$ be the set of prefixes $w \in \{0,1\}^{<n}$ that are not prefixes of $x$. In both $D_\mathrm{yes}$ and in $D_\mathrm{no}$, when conditioned on an explicitly given $x$, the random variables $\{s_{w}\}_{w \in I_x}$ are drawn uniformly and independently from $\{+1, -1\}$.
\end{lemma}
\begin{proof}
    In $D_\mathrm{yes}$, the independence of $x$ and $\{s_w\}_{w \in I_x}$ is trivial, since $x$ is drawn uniformly and independently of $\mu$.

    In $D_\mathrm{no}$, recall that $\mu'(x)$ only depends on $\{s_{x_{1,\ldots,i-1}}\}_{1\le i \le n}$:
    \[ \Pr_{\mu'}\left[x | \{s_w\}_{w \in \{0,1\}^{<n}}\right]
        = \prod_{i=1}^n \mu'(x_i | x_{1,\ldots,i-1})
        = \prod_{i=1}^n \frac{1 - s_{x_{1,\ldots,i-1}} r}{2} \]

    Hence, for every assignment $f : I_x \to \{+1,-1\}$, by the law of total probability,
    \begin{eqnarray*}
        \Pr\left[\forall w \in I_x : s_w = f(w) \cond x\right]
        &=& \!\!\!\! \sum_{g \in \{+1,-1\}^n} \!\!\!\! \Pr\left[\forall 1 \le i \le n : s_{x_{1,\ldots,i-1}} = g_i \cond x\right] \cdot \\ && \phantom{\sum_{g \in \{+1,-1\}^n}\!\!\!\! } \!\!\!\! \Pr\left[\forall w \in I_x : s_w = f(w) \cond \forall 1 \le i \le n : s_{x_{1,\ldots,i-1}} = g_i \right] \\
        (*) &=& \!\!\!\! \sum_{g \in \{+1,-1\}^n} \!\!\!\!  \Pr\left[\forall 1 \le i \le n : s_{x_{1,\ldots,i-1}} = g_i \cond x\right] \cdot \Pr\left[\forall w \in I_x : s_w = f(w)\right] \\
        &=& \!\!\!\! \Pr\left[\forall w \in I_x : s_w = f(w)\right]
    \end{eqnarray*}
    $(*)$: When not conditioned on $x$, the $s_w$s are independent.

    Therefore, the posterior distribution of $\{s_w\}_{w\in I_x}$, when conditioned on $x$, is the same as the prior distribution, which is uniform over $\{+1,-1\}^{I_x}$.
\end{proof}

\subsection{Concentration gap}

For $n \ge 2$, $\eps > 0$, $\delta = \sqrt{2}\eps/\sqrt{n}$ and $r = \paperZcoefZr / \sqrt{n}$, we define the following notations.

\begin{definition}[A few technical definitions] \label{def:few-technical-defs-for-lower-bound} ~
    \begin{align*}
        & k_H = \frac{1}{2}n - \sqrt{3}\sqrt{n} && k_L = \frac{1}{2}n - \sqrt{6}\sqrt{n} \\
        & p_H = 2^{-n} (1 + \delta)^{k_H} (1 - \delta)^{n - k_H} && p_L = 2^{-n} (1 + \delta)^{k_L} (1 - \delta)^{n - k_L} \\
        & H_\mu = \{ x : \mu(x) \ge p_H \} && L_\mu = \{ x : \mu(x) \le p_L \}
    \end{align*}
\end{definition}

The following lemmas (\ref{lemma:div-pH-pL}-\ref{lemma:dno-concentration-le-pL}) are technical. The ``front-end'' lemmas are \ref{lemma:div-pH-pL} and \ref{lemma:dyes-concentration-ge-pH}-\ref{lemma:dno-concentration-le-pL}. The rushing reader can skip the others (without missing core ideas) and continue from Page \pageref{sec:lbnd:subsec:lbnd-theorem}.

\begin{lemma} \label{lemma:div-pH-pL}
    For $n \ge 1$ and $\eps < 1/150$, $(1 - \eps)p_H > (1 + \eps)p_L$.
\end{lemma}
\begin{proof}
    \[  \frac{p_L}{p_H}
        = \frac{2^{-n} (1 + \delta)^{k_L} (1 - \delta)^{n - k_L}}{2^{-n} (1 + \delta)^{k_H} (1 - \delta)^{n - k_H}}
        = (1 - \delta)^{k_H - k_L} (1 + \delta)^{k_L - k_H}
        = \left(\frac{1 - \delta}{1 + \delta}\right)^{k_H - k_L}
        = \left(1 - \frac{2\delta}{1 + \delta}\right)^{k_H - k_L}
    \]

    We use a standard bound to obtain:
    \[  \frac{p_L}{p_H}
        \le \exp\left(-\frac{2\delta}{1+\delta}(k_H - k_L)\right)
        = \exp\left(-\frac{2\sqrt{2}\eps}{(1+\delta)\sqrt{n}} \cdot (\sqrt{6}-\sqrt{3})\sqrt{n}\right)
        \le e^{-2.029 \eps / (1+\eps)}
    \]

    Since $\frac{1+\eps}{1-\eps} = 1 + \frac{2\eps}{1-\eps} \le e^{2\eps/(1-\eps)}$, we obtain that $\frac{(1+\eps)p_L}{(1-\eps)p_H} \le \exp\left(\left(-\frac{2.029}{1+\eps} + \frac{2}{1-\eps}\right)\eps\right) < 1$, where the last transition holds for every $\eps < 1/150$.
\end{proof}

\begin{lemma} \label{lemma:mu-close-to-uniform}
    With probability $1$ over the choice of $\mu$, $\mu$ is $\eps$-close to the uniform distribution over $\{0,1\}^n$.
\end{lemma}
\begin{proof}
    When drawing $x$ bit by bit, the probability to choose $x_i=1$, given $x_1,\ldots,x_{i-1}$, is in the range $\frac{1 \pm \delta}{2}$. Therefore, $\DKL{\mu}{\mathcal U} \le n \cdot \dkl\left(\frac{1 \pm \delta}{2}, \frac{1}{2}\right) \le n \cdot \delta^2 = 2 \eps^2$, where the rightmost inequality transition is by Lemma \ref{lemma:cheatsheet:dkl-quadratic-bound}. By Pinsker's inequality, this implies that $\dtv(\mu,\mathcal U) \le \eps$.
\end{proof}

For the following lemmas we use Berry-Esseen theorem.

\begin{theorem}[Berry-Esseen]
    Let $X_1,\ldots,X_n$ be identically independently distributed variables with $\E[X_1] = 0$, $\Var[X_1] = \sigma^2$ and $\E[\abs{X_1}^3] = \rho$. Let $Y_n = \frac{1}{n} \sum_{i=1}^n X_i$. In this setting, for every $t \in \mathbb R$, $\abs{\Pr\left[\frac{Y_n\sqrt{n}}{\sigma} \le t\right] - \Phi(t)} \le \frac{\rho}{2\sigma^3 \sqrt{n}}$, where is the cumulative probability function of the standard normal distribution.
\end{theorem}

\begin{lemma}[Applying Berry-Esseen theorem on binomial distributions] \label{lemma:berry-esseen-modified}
    For every $\frac{1}{3} \le p \le \frac{2}{3}$ and $k \sim \Bin(n,p)$, it holds that $\Pr\left[\Bin(n,p) > np + \alpha \sqrt{n}\right] \in \Phi\left(-\alpha/\sqrt{p(1-p)}\right) \pm \frac{1}{\sqrt{n}}$ and $\Pr\left[\Bin(n,p) < np + \alpha \sqrt{n}\right] \in \Phi\left(\alpha/\sqrt{p(1-p)}\right) \pm \frac{1}{\sqrt{n}}$. In both cases, the strict inequality ($<$, $>$) can also be weak ($\le$, $\ge$).
\end{lemma}
\begin{proof}
    Let:
    \begin{itemize}
        \item $K_1,\ldots,K_n \sim \Ber(p)$ (iid).
        \item $K = \sum_{i=1}^n K_i$ (distributes like $\Bin(n,p)$).
        \item $X_i = K_i - p$.
        \item $Y_n = \frac{1}{n} X = \frac{1}{n}(K - np)$.
    \end{itemize}

    For Berry-Essen theorem, $\sigma = \sqrt{p(1-p)}$ and $\rho = p(1-p)^3 + (1-p)p^3 = \sigma^2 (p^2 + (1-p)^2)$. Hence,

    For the greater-than case,
    \begin{eqnarray*}
        \Pr\left[\Bin(np) > np + \alpha \sqrt{n} \right]
        &=& \Pr\left[K - np > \alpha \sqrt{n} \right] \\
        &=& \Pr\left[n Y_n > \alpha \sqrt{n} \right]
        = \Pr\left[\frac{Y_n\sqrt{n}}{\sigma} > \alpha/\sigma \right]
        = \Phi(\alpha/\sigma) \pm \frac{\rho}{2\sigma^3 \sqrt{n}}
    \end{eqnarray*}

    The analysis is the same for the less-than case, and analogous in the less-equals case and the greater-equals case as well. The additive error is bounded by:
    \[  \frac{\rho}{2\sigma^3 \sqrt{n}}
        = \frac{\sigma^2 (p^2 + (1-p)^2)}{2\sigma^3 \sqrt{n}}
        = \frac{p^2 + (1-p)^2}{2 \sigma \sqrt{n}}
        = \frac{p^2 + (1-p)^2}{2 \sqrt{p(1-p)}} \cdot \frac{1}{\sqrt{n}}
        \underset{1/3 \le p \le 2/3} \le \frac{1}{\sqrt{n}}
        \qedhere \]
\end{proof}

\begin{lemma} \label{lemma:mu-H}
    For $n \ge 2.2 \cdot 10^8$, with probability $1$ over the choice of $\mu$, $\mu(H_\mu) \ge 1 - e^{-8}$.
\end{lemma}
\begin{proof}
    Let $x \in \{0,1\}^n$ and $k(x)$ be the number of $\frac{1+\delta}{2}$-marginals (the number of indexes $1 \le i \le n$ for which $\mu(x_i | x_{1,\ldots,i-1}) = \frac{1+\delta}{2}$). Observe that $\mu(x)$ is fully defined by $k(x)$ and increasing monotone in $k$.

    When drawing $x$ from $\mu$, $k(x)$ is drawn from $\Bin\left(n,\frac{1+\delta}{2}\right)$, and hence:
    \[\Pr_\mu[x \notin H]
        = \Pr_\mu\left[k < k_H\right]
        = \Pr\left[k < \E[k] - \left(\sqrt{3}+\frac{1}{2}\delta\sqrt{n}\right)\sqrt{n}\right]
        \underset{\text{L \ref{lemma:berry-esseen-modified}}}\le \Phi\left(-\frac{\sqrt{3}+\frac{1}{\sqrt{2}}\eps}{\frac{1}{2}\sqrt{1-\delta^2}}\right) + \frac{1}{\sqrt{n}}
    \]

    By monotonicity of $\Phi$, Lemma \ref{lemma:lookup:Phi(-2sqrt(3))} and the bound $n \ge 2.2 \cdot 10^8$,
    \[ \Pr_\mu[x \notin H] \le \Phi(-2\sqrt{3}) + \frac{1}{\sqrt{n}}
        \le e^{-8.23} + \frac{1}{\sqrt{n}}
        \le e^{-8} \qedhere \]
\end{proof}

\begin{lemma} \label{lemma:mu-L-partA}
    For $n \ge 4.5 \cdot 10^{16}$ and $\eps < \frac{1}{2000}$, with probability $1$ over the choice of $\mu$, $\mu(L_\mu) > e^{-14.57}$.
\end{lemma}
\begin{proof}
    Let $x \in \{0,1\}^n$ and $k(x)$ be the number of $\frac{1+\delta}{2}$-marginals (the number of indexes $1 \le i \le n$ for which $\mu(x_i | x_{1,\ldots,i-1}) = \frac{1+\delta}{2}$). Observe that $\mu(x)$ and $\mu'(x)$ are fully defined by $k(x)$ and that both are monotone in $k(x)$ ($\mu$ increases, $\mu'$ decreases).

    When drawing $x$ from $\mu$, $k(x)$ is drawn from $\Bin\left(n,\frac{1+\delta}{2}\right)$, and hence:
    \begin{eqnarray*}
        \Pr_\mu[x \in L]
        = \Pr_\mu\left[k \le k_L\right]
        =\Pr\left[k \le \E[k] - (\sqrt{6} + \frac{1}{2}\delta\sqrt{n})\sqrt{n}\right]
        \underset{\text{L\ref{lemma:berry-esseen-modified}}}\ge \Phi\left(-\frac{\sqrt{6}+\frac{1}{\sqrt{2}}\eps}{\frac{1}{2}\sqrt{1-\delta^2}}\right) - \frac{1}{\sqrt{n}} 
    \end{eqnarray*}

    By monotonicity of $\Phi$ and the bounds of the parameters $n$, $\eps$ and $\delta=\sqrt{2}\eps/\sqrt{n} < 1/\sqrt{n}$,
    \begin{eqnarray*}
        \Pr_\mu[x \in L]
        \ge \Phi(-4.9) - \frac{1}{\sqrt{n}}
        \ge_{\text{L\ref{lemma:lookup:Phi(-4.9)}}} e^{-14.56} - \frac{1}{2.12 \cdot 10^8}
        \ge e^{-14.57}
    \end{eqnarray*}
\end{proof}

\begin{lemma} \label{lemma:mu-L-partB}
    For $n \ge 3 \cdot 10^{10}$ and $\eps > 0$, with probability $1$ over the choice of $\mu$ and $\mu'$, for every $x \in L_\mu$, $\mu'(x) \ge e^{7.19} \mu(x)$.
\end{lemma}
\begin{proof}
    Let $x \in \{0,1\}^n$ and $k(x)$ be the number of $\frac{1+\delta}{2}$-marginals (the number of indexes $1 \le i \le n$ for which $\mu(x_i | x_{1,\ldots,i-1}) = \frac{1+\delta}{2}$). Observe that $\mu(x)$ and $\mu'(x)$ are fully defined by $k(x)$ and that both are monotone in $k(x)$ ($\mu$ increases, $\mu'$ decreases).
    
    For $\mu'(x)$ with $k(x) \le k_L$:
    \begin{eqnarray*}
        2^n \mu'(x) &\ge& (1 - r)^{\frac{1}{2}n - \sqrt{6}\sqrt{n}} (1 + r)^{\frac{1}{2}n + \sqrt{6} \sqrt{n}} \\
        &=& (1 - r^2)^{\frac{1}{2}n} \left(\frac{1 - r}{1 + r}\right)^{-\sqrt{6} \sqrt{n}} \\
        &=& (1 - r^2)^{\frac{1}{2}n} \left(1 - \frac{2r}{1 + r}\right)^{-\sqrt{6}\sqrt{n}}
        \ge (1 - r^2)^{\frac{1}{2}n} \exp\left(\frac{2r}{1+r} \cdot \sqrt{6}\sqrt{n}\right) 
    \end{eqnarray*}

    We use $r=\paperZcoefZr/\sqrt{n}$ and $n \ge 3 \cdot 10^{10}$ to apply Lemma \ref{lemma:wolfram:(1-64/n)^(n/2)exp(16sqrt(6)/(1+8/sqrt(n)))>exp(7.19)}.
    \[  2^n \mu'(x)
        \ge (1 - \paperZcoefZrZsquared/n)^{\frac{1}{2}n} \exp\left(\frac{2\cdot \paperZcoefZr/\sqrt{n}}{1+\paperZcoefZr/\sqrt{n}} \cdot \sqrt{6} \sqrt{n}\right)
        \ge e^{7.19}
    \]
    
    For $\mu(x)$ with $k(x) \le k_L$:
    \[  2^n \mu(x) \le (1 + \delta)^{\frac{1}{2}n - \sqrt{6} \sqrt{n}} (1 - \delta)^{\frac{1}{2}n + \sqrt{6} \sqrt{n}}
        = (1 - \delta^2)^{\frac{1}{2}n} \left(\frac{1 - \delta}{1 + \delta}\right)^{\sqrt{6} \sqrt{n}}
        < 1 \]

    Combined, $\mu'(x)/\mu(x) > e^{7.19}$ for every $x \in L_\mu$.
\end{proof}

\begin{lemma} \label{lemma:dyes-concentration-ge-pH}
    For $n \ge 4.5 \cdot 10^{16}$ and $\eps \le e^{-11}$, $\Pr_{(\mu,x)\sim D_\mathrm{yes}}[\mu(x) \ge p_H] \ge 1 - e^{-7.95}$.
\end{lemma}
\begin{proof}
    Let $\mathcal U$ be the uniform distribution over $\{0,1\}^n$. By the triangle inequality, $\mathcal U(H_\mu) \ge \mu(H_\mu) - \dtv(\mu,\mathcal U)$. By Lemma \ref{lemma:mu-close-to-uniform} and Lemma \ref{lemma:mu-H}, this is at least $1 - e^{-8} - \eps \ge 1 - e^{-7.95}$. Therefore,
    \[  \Pr_{(\mu,x)\sim D_\mathrm{yes}}[\mu(x) \ge p_H]
        = \E\nolimits_{\text{construction of $\mu$}}[\mathcal U(H_\mu)]
        \ge 1 - e^{-7.95}
        \qedhere
    \]
\end{proof}

\begin{lemma} \label{lemma:dno-concentration-le-pL}
    For $n \ge 4.5 \cdot 10^{16}$ and $\eps \le e^{-11}$, $\Pr_{(\mu,x)\sim D_\mathrm{no}}[\mu(x) \le p_L] \ge e^{-7.38}$.
\end{lemma}
\begin{proof}
    By Lemma \ref{lemma:mu-L-partA} and Lemma \ref{lemma:mu-L-partB}, with probability $1$,
    \[\mu'(L_\mu)
        \ge \min_{x \in L_\mu} \frac{\mu'(x)}{\mu(x)} \cdot \mu(L_\mu)
        \ge e^{7.19} \cdot e^{-14.57}
        \ge e^{-7.38} \]

    Therefore,
    \[ \Pr_{(\mu,x) \sim D_\mathrm{no}}[\mu(x) \le p_L]
        = \E\nolimits_{\text{construction of $(\mu,\mu')$}}[\mu'(L_\mu)]
        \ge e^{-7.38}
    \qedhere \]
\end{proof}

\subsection{The lower bound theorem}
\label{sec:lbnd:subsec:lbnd-theorem}

\begin{lemma} \label{lemma:Alg-Dyes-Dno-large-distance}
    For $n \ge 4.5 \cdot 10^{16}$ and $\eps \le e^{-11}$, let $\mathcal A$ be an algorithm that uses prefix conditional samples and results in a real number such that for every distribution $\mu$ over $\Omega = \{0,1\}^n$ there exists a set $G_\mu \subseteq \Omega$ of mass $\Pr_\mu[G_\mu] > 1 - \paperZGmuZerr$ for which $\Pr\left[\mathcal A(\mu,x) \in (1 \pm \eps)\mu(x)\right] > 1 - \paperZGmuZprZestZerr$ for every $x \in G_\mu$. In this setting, the total variation distance of the behaviors of the algorithm when given an input drawn from $D_\mathrm{yes}$ and from $D_\mathrm{no}$ is greater than $e^{-8.25}$.
\end{lemma}

Note that, for estimation algorithms, having $\mu(G_\mu) = 1$ is infeasible, and general upper bound algorithms usually guarantee that $\mu(G_\mu) = 1 - O(\eps)$. Here we show a lower bound for a much weaker demand, that allows for excluding a fixed weight of the elements.

\begin{proof}
    Consider the distributions $D_\mathrm{yes}$ and $D_\mathrm{no}$ and the following algorithm: we use $\mathcal A$ to estimate $(1 \pm \eps)\mu(x)$. If the result estimation is greater than $(1-\eps)p_H$ then we accept, and otherwise we reject. By Lemma \ref{lemma:div-pH-pL}, the ranges $(1 \pm \eps)p_H$ and $(1 \pm \eps)p_L$ are disjoint, and therefore, a $1 \pm \eps$-multiplicative error suffices to distinguish between $p_H$ and $p_L$.

    Let $B_\mathrm{wrong}$ be the bad event ``the estimation of $\mu(x)$ is outside the range $(1 \pm \eps) \mu(x)$''. The probability to reject an input $(\mu,x)$ drawn from $D_\mathrm{yes}$ is bounded by:
    \begin{eqnarray*}
        \Pr\left[\reject \cond D_\mathrm{yes}\right]
        &\le& \Pr_{D_\mathrm{yes}}[x \notin H_\mu \cap G_\mu] + \Pr[B_\mathrm{wrong} | x \in G_\mu] \\
        &\le& \dtv(\mathcal U, \mu) + (1 - \mu(H_\mu)) + (1 - \mu(G_\mu)) + \paperZGmuZprZestZerr \\
        \text{[L\ref{lemma:mu-close-to-uniform}, L\ref{lemma:mu-H}]} &\le& \eps + e^{-8} + \paperZGmuZerr + \paperZGmuZprZestZerr
    \end{eqnarray*}

    By Lemma \ref{lemma:mu-L-partA} and Lemma \ref{lemma:mu-L-partB},
    \begin{eqnarray*}
        \mu'(L_\mu \cap G_\mu)
        \ge \min_{x \in L_\mu} \frac{\mu'(x)}{\mu(x)} \cdot (\mu(L_\mu) - (1 - \mu(G_\mu)))
        \ge e^{7.19} \cdot (e^{-14.57} - \paperZGmuZerr)
        \ge e^{-7.385}
    \end{eqnarray*}

    The probability to reject an input $(\mu,x)$ drawn from $D_\mathrm{no}$ is at least:
    \begin{eqnarray*}
        \Pr\left[\reject \cond D_\mathrm{no}\right]
        \ge \Pr_{D_\mathrm{no}}[x \in G_\mu \cap L_\mu] \cdot \Pr[\neg B_\mathrm{wrong} | x \in G_\mu]
        \ge e^{-7.385} \cdot (1 - \paperZGmuZprZestZerr)
        \ge e^{-7.39}
    \end{eqnarray*}
    
    Therefore, the distance in behaviors of the algorithms when running on $D_\mathrm{yes}$ and $D_\mathrm{no}$ is at least $e^{-7.39} - (\eps + e^{-8}+ \paperZGmuZerr + \paperZGmuZprZestZerr) > e^{-8.25}$ (for $\eps \le e^{-11}$).
\end{proof}

\begin{lemma} \label{lemma:lower-bound-prefix-marginal-access}
    For $n \ge 4.5 \cdot 10^{16}$ and $\eps < e^{-11}$, let $\mathcal A$ be an algorithm that uses marginal prefix conditional samples and results in a real number such that for every distribution $\mu$ over $\Omega = \{0,1\}^n$ there exists a set $G_\mu \subseteq \Omega$ of mass $\Pr_\mu[G_\mu] > 1 - \paperZGmuZerr$ for which $\Pr\left[\mathcal A(\mu,x) \in (1 \pm \eps)\mu(x)\right] > 1 - \paperZGmuZprZestZerr$ for every $x \in G_\mu$. In this setting, $\mathcal A$ must make at least $\frac{n^2}{\paperZanticoefZfinalqZbyZnsqrZoverZepssqr \eps^2}$ marginal prefix conditional samples that contain $x$ in expectation.
\end{lemma}
\begin{proof}
    Observe that, for every $w \in \{0,1\}^{<n}$, the marginal prefix conditional sample $\mu|_{\abs{w}+1}^{w \times \{0,1\}^{n-\abs{w}}}$ is identical to drawing a bit from $\Ber\left(\frac{1+s_w \delta}{2}\right)$. By Lemma \ref{lemma:bit-by-bit-rest-is-independent}, the random variables $\{s_w\}_{w \in I_x}$, for $I_x = \{0,1\}^{<n} \setminus \{ 1 \le i \le n : x_{1,\ldots,i-1}\}$, are independent, and distribute the same in $D_\mathrm{yes}$ and in $D_\mathrm{no}$. Hence, sampling them is redundant.
    
    In other words, in every non-redundant step, the algorithm chooses some $1 \le i \le n$ and samples a bit whose probability to be $1$ is $\frac{1 + s_{x_{1,\ldots,i-1}} \delta}{2}$. For every $1 \le i \le n$, let $p_i = \frac{1+s_{x_{1,\ldots,i-1}} \delta}{2}$.

    In $D_\mathrm{yes}$, every $s_w$ is uniformly drawn from $\{+1,-1\}$, and hence, the $p_i$s are uniformly drawn from $\frac{1 \pm \delta}{2}$. This matches the $D_n(\delta,0)$ distribution of the sequence-distinguishing task. In $D_\mathrm{no}$, when conditioned on drawing $x$ from $\mu'$, every $s_w$ (for $w$ on the path of $x$) is $+1$ with probability $\frac{1-r}{2}$ and otherwise $-1$. This matches the $D_n(\delta,r)$ distribution of the sequence-distinguishing task.

    Note that this analysis is oblivious of the dependence between the $s_w$s and $x$. Observe that if we denote $s_{x_{1,\ldots,i-1}}$ by $\hat{s}_{i-1}$, then the sequence $(\hat{s}_0,\ldots,\hat{s}_{n-1})$ is independent of $x$ and iid.

    Let $q$ be the expected number of non-redundant steps. By Markov's inequality, with probability at least $1-e^{-9}$ it is bounded by $e^9 q$.

    If $e^9 q \le \frac{(e^{-9})^4}{\paperZanticoefZqZbyZnZoverZdeltasqr \delta^2} n = \frac{n^2}{e^{36} \cdot \paperZanticoefZqZbyZnZoverZdeltasqr \cdot 2 \eps^2}$, then by Lemma \ref{lemma:lower-bound-ad-hoc-adaptive}, the total variation distance between the runs is less than $e^{-9}$.

    Overall, if $q \le \frac{n^2}{\paperZanticoefZfinalqZbyZnsqrZoverZepssqr \eps^2} \le \frac{n^2}{e^9 \cdot e^{36} \cdot \paperZanticoefZqZbyZnZoverZdeltasqr \cdot 2 \eps^2}$, then the total-variation distance between the run is bounded by $e^{-9} + e^{-9} < e^{-8.25}$. By Lemma \ref{lemma:Alg-Dyes-Dno-large-distance}, a ``good'' estimation algorithm behaves more differently ($\dtv > e^{-8.25}$) on $D_\mathrm{yes}$ on $D_\mathrm{no}$, and hence, every such an algorithm must make at least $\frac{n^2}{\paperZanticoefZfinalqZbyZnsqrZoverZepssqr \eps^2}$ prefix conditional queries.
\end{proof}

We describe a setting for the following lemma.
\begin{definition}[Setting for Lemma \ref{lemma:expected-number-of-effective-marginal-samples}] \label{def:for-lemma-expected-number-of-effective-marginal-samples}
    Let $x \in \{0,1\}^n$, $1 \le i \le n$ and $w \in \{0,1\}^{i-1}$ be a prefix. Let $X_i,\ldots,X_n$ be the bits of a prefix conditional sample $\mu|^{w\times \{0,1\}^{n-i}}$. Let $W_{i-1},\ldots,W_{n-1}$ be the intermediate prefixes ($W_{i-1}=w$, $W_{j-1} = W_{j-2}X_{j-1}$ for $i+1 \le j \le n$). Also, let $N_\mathrm{effective} = \abs{\{ i \le j \le n : \text{$W_{j-1}$ is a prefix of $x$}\}}$.
\end{definition}

As observed before, the sampling from $\mu|^{w\times \{0,1\}^{n-i}}$ can be decomposed into $n-i+1$ marginal samplings such that for $i \le j \le n$, after $W_{i-1}$ is determined, $X_j$ distributes as $\mu|_j^{W_{j-1} \times \{0,1\}^{n-j}}$. Therefore, $N_\mathrm{effective}$ describes the number of effective marginal samples.

\begin{lemma} \label{lemma:expected-number-of-effective-marginal-samples}
      In the setting of Definition \ref{def:for-lemma-expected-number-of-effective-marginal-samples}, if $n \ge 20$, then $\E[N_\mathrm{effective}] \le 3$.
\end{lemma}
\begin{proof}
    Note that $n \ge 20$ implies that $\delta=\sqrt{2}\eps/\sqrt{n} \le \sqrt{1/10} < 1/3$.
    
    If $w$ is not a prefix of $x$, then none of $W_{i-1},\ldots,W_{n-1}$ can be a prefix of $x$, and hence $\E[N_\mathrm{effective}] = 0 \le 3$. In the rest of the proof we assume that $w$ is a prefix of $x$.
    
    For $i+1 \le j \le n$, observe that if $W_{j-1}$ is a prefix of $x$ then $W_{j-2}$ is a prefix of $x$ as well. More specifically,
    \begin{eqnarray*}
        \Pr\left[\text{$W_{j-1}$ is a prefix of $x$}\right]
        &=& \Pr\left[X_{j-1} = x_{j-1} \cond W_{j-2} = x_{1,\ldots,j-2} \right] \Pr\left[\text{$W_{j-2}$ is a prefix of $x$}\right] \\
        \text{[Repeated application]} &=& \prod_{k=i+1}^j \Pr\left[X_{k-1} = x_{k-1} \cond W_{k-2} = x_{1,\ldots,k-2} \right] \cdot \Pr\left[\text{$W_{i-1}$ is a prefix of $x$}\right] \\
        &=&\prod_{k=i+1}^j \mu(x_{k-1} | x_{1,\ldots,k-2}) \cdot 1
        \le \left(\frac{1 + \delta}{2}\right)^{j-i}
        \le \left(\frac{2}{3}\right)^{j-i}
    \end{eqnarray*}

    By linearity of expectation,
    \[  \E[N_\mathrm{effective}]
        = \sum_{j=i}^n \E\left[\mathbf 1 \left[\text{$W_{j-1}$ is a prefix of $x$}\right]\right]
        \le 1 + \sum_{j=i+1}^n \left(\frac{2}{3}\right)^{j-i}
        < \sum_{t=0}^\infty \left(\frac{2}{3}\right)^t
        = 3 \qedhere \]
\end{proof}

\begin{theorem} \label{th:lower-bound-prefix-new}
    For $n \ge 4.5 \cdot 10^{16}$ and $\eps \le e^{-11}$, let $\mathcal A$ be an algorithm that uses prefix conditional samples and results in a real number such that for every distribution $\mu$ over $\Omega = \{0,1\}^n$ there exists a set $G_\mu \subseteq \Omega$ of mass $\Pr_\mu[G_\mu] > 1 - \paperZGmuZerr$ for which $\Pr\left[\mathcal A(\mu,x) \in (1 \pm \eps)\mu(x)\right] > 1 - \paperZGmuZprZestZerr$ for every $x \in G_\mu$. In this setting, $\mathcal A$ must use $\Omega(n^2/\eps^2)$ prefix conditional samples.
\end{theorem}
\begin{proof}
    Recall that we can see every prefix condition sampling with $k$ free bits as a (random) sequence of $k$ marginal prefix condition sampling. In the $i$th step, we condition on the first already-fixed $n-k$ bits and on the additional $i-1$ already-sampled bits to obtain the $n-k+i$th bit. By Lemma \ref{lemma:expected-number-of-effective-marginal-samples}, for every individual prefix sample, the expected number of marginal samples whose condition contains $x$ is bounded by $3$. By linearity of expectation, the expected number of marginal samples whose prefix condition contains $x$ is bounded by $3q$.

    By Lemma \ref{lemma:lower-bound-prefix-marginal-access}, since the expected number of ``non-redundant'' samples is $O(q)$, the algorithm must make $q = \Omega(n^2/\eps^2)$ prefix conditional samples.
\end{proof}

\newpage
\appendix

\section{Lower bound for the sequence-distinguishing task}
\label{apx:ad-hoc-task}

In this appendix we show a lower bound for the sequence-distinguishing task. The goal is distinguishing between the following random sequences, regarding the parameters $\delta > 0$ and $r \ge 0$:
\begin{itemize}
    \item $D_n(\delta,0)$: $p_i = (1 + \delta)/2$ with probability $1/2$, and otherwise $p_i = (1 - \delta)/2$, independently for every $1 \le i \le n$.
    \item $D_n(\delta,r)$: $p_i = (1 + \delta)/2$ with probability $(1-r)/2$, and otherwise $p_i = (1 - \delta)/2$, independently for every $1 \le i \le n$.
\end{itemize}
In each step, the algorithm can choose an index $1 \le i \le n$ and obtain a bit drawn from $\Ber(p_i)$.


We show that, if $n \le \paperZcoefZrZsquared/r^2$, then we must make $\Omega(1/r^2 \delta^2)$ queries, even if we allow adaptive behavior. 

\subsection{Technical lemmas}

We first state a few technical lemmas.

\begin{lemma} \label{lemma:bounded-ratio-bounded-dkl}
    Let $0 \le t \le 1/4$. For two distributions $\mu$ and $\nu$ over $\Omega$ for which $\mu(x) \in (1 \pm t)\nu(x)$ for every $x \in \Omega$, $\DKL{\mu}{\nu} \le t^2 / \ln 2$.
\end{lemma}
\begin{proof}
    For every $x \in \Omega$ for which $\mu(x) \ge \nu(x)$, we can use $-\ln(1+x) \le -x + x^2$ (Lemma \ref{lemma:cheatsheet:neg-ln-1+x-le-neg-x-plus-x-sqr}), which holds for $\abs{x} \le 2/3$:
    \begin{eqnarray*}
        -\ln \frac{\nu(x)}{\mu(x)}
        = -\ln \left(1 + \frac{\nu(x) - \mu(x)}{\mu(x)}\right)
        &\le& -\frac{\nu(x) - \mu(x)}{\mu(x)} + \frac{(\nu(x) - \mu(x))^2}{(\mu(x))^2} \\
        &=& -\frac{\nu(x) - \mu(x)}{\mu(x)} + \left(\frac{\nu(x)}{\mu(x)} - 1\right)^2 \\
        \text{[Since $\nu(x) \le \mu(x) \le (1+t)\nu(x)$]} &\le& -\frac{\nu(x) - \mu(x)}{\mu(x)} + \left(\frac{1}{1+t} - 1\right)^2
        \le -\frac{\nu(x) - \mu(x)}{\mu(x)} + t^2
    \end{eqnarray*}
    
    For every $x \in \Omega$ for which $\mu(x) \le \nu(x)$, we can use $-\ln (1+x) \le -x + x^2/2$ (Lemma \ref{lemma:cheatsheet:neg-ln-1+x-le-neg-x-plus-x-sqr-half}):
    \begin{eqnarray*}
        -\ln \frac{\nu(x)}{\mu(x)}
        = -\ln \left(1 + \frac{\nu(x) - \mu(x)}{\mu(x)}\right)
        &\le& -\frac{\nu(x) - \mu(x)}{\mu(x)} + \frac{(\nu(x) - \mu(x))^2}{2(\mu(x))^2} \\
        &=& -\frac{\nu(x) - \mu(x)}{\mu(x)} + \frac{1}{2} \left(\frac{\nu(x)}{\mu(x)} - 1\right)^2 \\
        \text{[Since $(1-t) \nu(x) \le \mu(x) \le \nu(x)$]} &\le& -\frac{\nu(x) - \mu(x)}{\mu(x)} + \frac{1}{2} \left(\frac{1}{1-t} - 1\right)^2
        \le -\frac{\nu(x) - \mu(x)}{\mu(x)} + t^2
    \end{eqnarray*}
    (In the last transition we used the constraint $t \le 1/4$).

    Combined,
    \begin{eqnarray*}
        \ln 2 \cdot \DKL{\mu}{\nu}
        &=& -\sum_{x \in \Omega} \mu(x) \ln \frac{\nu(x)}{\mu(x)} \\
        &\le& \sum_{x \in \Omega} (-(\nu(x) - \mu(x)) + \mu(x) t^2)
        = -\sum_{x \in \Omega} (\nu(x) - \mu(x)) + t^2
        = 0 + t^2
        = t^2
    \end{eqnarray*}
\end{proof}

\begin{lemma} \label{lemma:symmetric-chi-square-bounded-by-sum-dkls} 
    Let $\mu$ and $\nu$ be two distributions over a domain $\Omega$. It holds that:
    \[ \sum_{x \in \Omega} \frac{(\mu(x) - \nu(x))^2}{\mu(x) + \nu(x)} \le (\DKL{\mu}{\nu} + \DKL{\nu}{\mu}) \cdot \ln 2 \]
\end{lemma}
\begin{proof}
    If $\mu$ and $\nu$ are supported on different sets, then $\DKL{\mu}{\nu} + \DKL{\nu}{\mu}$ is infinite, and trivially greater than the finite left-hand expression. We proceed under the assumption that $\mu$ and $\nu$ have the same support set.
    
    Let $A = \{ x \in \Omega : \mu(x) > \nu(x) \}$ and $B = \{ x \in \Omega : \mu(x) < \nu(x) \}$.
    \begin{eqnarray*}
        \DKL{\mu}{\nu} + \DKL{\nu}{\mu}
        &=& \sum_{x \in \Omega} \mu(x) \log \frac{\mu(x)}{\nu(x)} + \sum_{x \in \Omega} \nu(x) \log \frac{\nu(x)}{\mu(x)} \\
        &=& -\sum_{x \in A} (\mu(x) - \nu(x)) \log \frac{\nu(x)}{\mu(x)} - \sum_{x \in B} (\nu(x) - \mu(x)) \log \frac{\mu(x)}{\nu(x)}
    \end{eqnarray*}

    For the $A$-series, we use $-\ln(1+x) \ge -x$ (Lemma \ref{lemma:cheatsheet:neg-ln-1+x-ge-neg-x}) to obtain that:
    \begin{eqnarray*}
        -\sum_{x \in A} (\mu(x) - \nu(x)) \log \frac{\nu(x)}{\mu(x)}
        &=& -\sum_{x \in A} (\mu(x) - \nu(x)) \log \left(1 + \frac{\nu(x) - \mu(x)}{\mu(x)}\right) \\
        \text{[Since $\mu(x)-\nu(x) > 0$]} &\ge& -\frac{1}{\ln 2}\sum_{x \in A} (\mu(x) - \nu(x)) \frac{\nu(x) - \mu(x)}{\mu(x)} \\
        &=& \frac{1}{\ln 2}\sum_{x \in A} \frac{(\mu(x) - \nu(x))^2}{\mu(x)}
        \ge \frac{1}{\ln 2}\sum_{x \in A} \frac{(\mu(x) - \nu(x))^2}{\mu(x) + \nu(x)}
    \end{eqnarray*}

    Analogously, for the $B$-series,
    \begin{eqnarray*}
        -\sum_{x \in B} (\nu(x) - \mu(x)) \log \frac{\mu(x)}{\nu(x)}
        &=& -\sum_{x \in B} (\nu(x) - \mu(x)) \log \left(1 + \frac{\mu(x) - \nu(x)}{\nu(x)}\right) \\
        \text{[Since $\nu(x)-\mu(x) > 0$]} &\ge& -\frac{1}{\ln 2}\sum_{x \in B} (\nu(x) - \mu(x)) \frac{\mu(x) - \nu(x)}{\nu(x)} \\
        &=& \frac{1}{\ln 2}\sum_{x \in B} \frac{(\mu(x) - \nu(x))^2}{\nu(x)}
        \ge \frac{1}{\ln 2}\sum_{x \in B} \frac{(\mu(x) - \nu(x))^2}{\mu(x) + \nu(x)}
    \end{eqnarray*}

    Combined,
    \[  \DKL{\mu}{\nu} + \DKL{\nu}{\mu}
        \ge \frac{1}{\ln 2} \sum_{x \in \Omega} \frac{(\mu(x) - \nu(x))^2}{\mu(x) + \nu(x)} \qedhere \]
\end{proof}

\begin{lemma} \label{lemma:dkl-of-half-vs-half-and-r-bias}
    Let $\mu$ and $\nu$ be two distributions over a domain $\Omega$. If $\abs{r} < \frac{1}{2}$, then:
    \[ \DKL{\frac{1}{2}\mu + \frac{1}{2}\nu}{\frac{1+r}{2}\mu + \frac{1-r}{2}\nu} \le \frac{1}{2}r^2 (\DKL{\mu}{\nu} + \DKL{\nu}{\mu}) \]
\end{lemma}
\begin{proof}
    \begin{eqnarray*}
        \DKL{\frac{1}{2}\mu + \frac{1}{2}\nu}{\frac{1+r}{2}\mu + \frac{1-r}{2}\nu}
        &=& -\sum_{x \in \Omega} \frac{1}{2}(\mu(x) + \nu(x)) \log \frac{\frac{1+r}{2}\mu(x) + \frac{1-r}{2}\nu(x)}{\frac{1}{2}\mu(x) + \frac{1}{2}\nu(x)} \\
        &=& -\sum_{x \in \Omega} \frac{1}{2}(\mu(x) + \nu(x)) \log \left(1 + r \frac{\mu(x) - \nu(x)}{\mu(x) + \nu(x)}\right)
    \end{eqnarray*}
    
    Since $\abs{r} < 1/2$ and $\abs{\frac{\mu(x) - \nu(x)}{\mu(x) + \nu(x)}} \le 1$, we can use the bound $-\ln (1+x) \le -x+x^2$ (Lemma \ref{lemma:cheatsheet:neg-ln-1+x-le-neg-x-plus-x-sqr}) to obtain:
    \begin{eqnarray*}
        [\cdots] &\le& \frac{1}{\ln 2}\sum_{x \in \Omega} \frac{1}{2}(\mu(x) + \nu(x)) \left(-r \frac{\mu(x) - \nu(x)}{\mu(x) + \nu(x)} + r^2 \frac{(\mu(x) - \nu(x))^2}{(\mu(x) + \nu(x))^2}\right) \\
        &=& -\frac{1}{2 \ln 2} r \underbrace{\sum_{x \in \Omega} (\mu(x) - \nu(x))} + \frac{1}{2 \ln 2}r^2 \sum_{x \in \Omega} \frac{(\mu(x) - \nu(x))^2}{\mu(x) + \nu(x)} \\
        &=& \phantom{-\frac{1}{2\ln 2}r \sum_{x \in \Omega} (\mu(x} 0 \phantom{) - \nu(x))} \!\!\!+ \frac{1}{2 \ln 2}r^2 \sum_{x \in \Omega}  \frac{(\mu(x) - \nu(x))^2}{\mu(x) + \nu(x)} \\
        \text{[Lemma \ref{lemma:symmetric-chi-square-bounded-by-sum-dkls}]} &\le& \frac{1}{2}r^2 \left(\DKL{\mu}{\nu} + \DKL{\nu}{\mu}\right)
    \end{eqnarray*}
\end{proof}

\subsection{Lower bound for non-adaptive algorithms}
Before we introduce the lower bound against adaptive algorithms, we state a lower bound for non-adaptive algorithms. A non-adaptive algorithm draws $m_i$ samples from the $i$th index for every $1 \le i \le n$. The sample complexity is $q = \sum_{i=1}^n m_i$.

\begin{lemma} \label{lemma:ad-hoc-lower-bound-non-adaptive}
    For $\delta<1/3$, consider an algorithm that for every $i \ge 1$ draws $m_i$ samples from the $i$th index, and let $q = \sum_{i=1}^n m_i$. The KL-divergence of runs over inputs drawn from $D_n(\delta,0)$ and from runs over inputs drawn from $D_n(\delta,r)$ is bounded by $5 r^2 \delta^2 q$.
\end{lemma}
\begin{proof}
    A \emph{run} of the non-adaptive algorithm over an input is a pseudo-matrix $M \in \prod_{i=1}^n \{0,1\}^{m_i}$ consisting of the result bits. Since the matrix columns distribute binomially and independently, we can reduce the run matrix into a run sequence $M \in \prod_{i=1}^n \{0,\ldots,m_i\}$, where the $i$th entry corresponds to the number of $1$s in the $i$th column.

    Consider the following four distributions:
    \begin{align*}
        & \tau^{(i)}_- : \Bin\left(m_i, \frac{1-\delta}{2}\right), &&
        \tau^{(i)}_+ : \Bin\left(m_i, \frac{1+\delta}{2}\right) \\
        & \mu^{(i)} : \frac{1}{2}\tau^{(i)}_- + \frac{1}{2}\tau^{(i)}_+, &&
        \nu^{(i)} : \frac{1+r}{2}\tau^{(i)}_- + \frac{1-r}{2}\tau^{(i)}_+
    \end{align*}
    
    Clearly, A run over inputs drawn from $D_n(\delta,0)$ distributes like $\prod_{i=1}^n \mu^{(i)}$ and a run over inputs drawn from $D_n(\delta,r)$ distributes like $\prod_{i=1}^n \nu^{(i)}$. By Lemma \ref{lemma:cheatsheet:dkl-of-product}, the KL-divergence of runs is bounded by:
    \begin{eqnarray*}
        \DKL{\prod_{i=1}^n \mu^{(i)}}{\prod_{i=1}^n \nu^{(i)}}
        &=& \sum_{i=1}^n \DKL{\mu^{(i)}}{\nu^{(i)}} \\
        \text{[Lemma \ref{lemma:dkl-of-half-vs-half-and-r-bias}]} &\le& \sum_{i=1}^n \frac{1}{2}r^2 \left(\DKL{\tau^{(i)}_-}{ \tau^{(i)}_+} + \DKL{\tau^{(i)}_+}{\tau^{(i)}_-}\right) \\
        \text{[Lemma \ref{lemma:cheatsheet:dkl-bin-m-p-m-q}]} &=& \frac{1}{2}r^2 \sum_{i=1}^n m_i \cdot \left(\dkl\left(\frac{1-\delta}{2}, \frac{1+\delta}{2}\right) + 
        \dkl\left(\frac{1+\delta}{2}, \frac{1-\delta}{2}\right)\right) \\
        \text{[Lemma \ref{lemma:cheatsheet:dkl-quadratic-bound}]} &\le& \frac{1}{2}r^2 \left(\sum_{i=1}^n m_i\right) \cdot \left(\frac{\left(\frac{1 - \delta}{2} - \frac{1 + \delta}{2}\right)^2}{\frac{1+\delta}{2} \frac{1-\delta}{2}} + \frac{\left(\frac{1 + \delta}{2} - \frac{1 - \delta}{2}\right)^2}{\frac{1-\delta}{2} \frac{1+\delta}{2}} \right)
        \\ 
        &=& \frac{1}{2}r^2 q \cdot \left(\frac{4\delta^2}{1-\delta^2} + \frac{4\delta^2}{1-\delta^2}\right)
        = \frac{4}{1 - \delta^2} r^2 \delta^2 q 
        \underset{\delta<1/3}\le 5 r^2 \delta^2 q
    \end{eqnarray*}
\end{proof}

\subsection{Lower bound for adaptive algorithms}
Adaptive algorithms are allowed to choose their next step based on past samples. Therefore, their analysis is more difficult. By Yao's principle \cite{yao77}, a probabilistic algorithm is equivalent to a distribution over deterministic algorithms. Moreover, the success probability of a probabilistic algorithm is the expected success probability over the choice of a deterministic algorithm.

Every deterministic adaptive algorithm can be described as a decision tree, where edges correspond to samples ($0$ and $1$) and every non-leaf node is associated with the index $1 \le i \le n$ to sample in the following step. A run of a deterministic algorithm is its execution path, which corresponds to the leaf the algorithm reaches in its last step.

To show a lower bound of $q$ samples, we have to show that for every decision tree of height $q$, the total-variation distance between the distribution of leaves reached when the input is drawn from $D_n(\delta,0)$ and their distribution when the input is drawn from $D_n(\delta,r)$ is bounded by some $d < 1/3$. This implies that the success probability, which cannot exceed the maximal success probability of a deterministic algorithm, is at most $1-d$. By Pinsker's inequality, it suffices to bound the KL-divergence of the leaf distributions by $2d^2$.

We repeatedly apply the chain-rule for KL-divergence (Lemma \ref{lemma:cheatsheet:chain-rule-dkl}) to obtain that the KL-divergence of the leaf distributions is bounded by the expected KL-divergence of every step individually.

To analyze an adaptive algorithm, whose structure can be arbitrarily detailed, we simulate it using a stronger model whose runs have shorter descriptions. Let $m = \floor{\frac{d^2}{\paperZanticoefZmZbyZrecpZdeltasqr \delta^2}}$ and $q \le \frac{d^2}{\paperZanticoefZqZbyZmn}mn$. At the initialization phase, we run an $m$-height rectangle algorithm, which is a non-adaptive algorithm that makes $m$ queries in each index $1 \le i \le n$, obtaining a pre-processing sampling matrix $M \in \{0,1\}^{m\times n}$ at the cost of $mn$ samples. For every individual index $1 \le i \le n$, the first $m$ queries are answered based on the pre-processing sampling matrix. Once the algorithm makes the $m+1$st sample on the $i$th index, the simulator reveals the exact value of $p_i$, and exactly simulates further queries to the $i$th index using the full-knowledge about $p_i$.

\begin{observation} \label{obs:markov-high-columns}
    Let $d > 0$. An adaptive algorithm that draws $q \le \frac{d^2}{\paperZanticoefZqZbyZmn} mn$ samples has at most $\frac{d^2}{\paperZanticoefZqZbyZmn} n$ indexes sampled more than $m$ times.
\end{observation}

\begin{observation}
    Every run of an adaptive algorithm with $q \le \frac{d^2}{\paperZanticoefZqZbyZmn} mn$ samples can be described as a tuple of the pre-processing sampling matrix $M$, the set $I$ of at most $\frac{d^2}{\paperZanticoefZqZbyZmn} n$ fully-revealed indexes and a map $f : I \to \{+,-\}$ corresponding to the exact value of $p_i \in \frac{1\pm \delta}{2}$.
\end{observation}

\begin{lemma} \label{lemma:lower-bound-ad-hoc-adaptive}
    Let $\delta<1/3$, $r < 1/12$, $n \le \paperZcoefZrZsquared/r^2$ and $d > 0$. For every adaptive algorithm with $q \le \frac{d^4}{\paperZanticoefZqZbyZnZoverZdeltasqr \delta^2}n$ samples, the distance between runs over inputs drawn from $D_n(\delta,0)$ and runs over inputs drawn from $D_n(\delta,r)$ is bounded by $d$.
\end{lemma}
We believe that Lemma \ref{lemma:lower-bound-ad-hoc-adaptive} holds for all $n$ (possibly up to a constant factor), but $n\le \paperZcoefZrZsquared/r^2$ suffices for the main result of the paper.
\begin{proof}
    The KL-divergence of the algorithm runs cannot be higher than the KL-divergence of the simulator runs, by the information processing inequality (Lemma \ref{lemma:cheatsheet:information-processing-inequality}).

    Let $m = \floor{\frac{d^2}{\paperZanticoefZmZbyZrecpZdeltasqr \delta^2}}$. The initialization phase of the simulator consists of $m$ samples in every index $1 \le i \le n$, and hence by Lemma \ref{lemma:ad-hoc-lower-bound-non-adaptive}, the KL-divergence of this phase for $D_n(\delta,0)$ and $D_n(\delta,r)$ is bounded by $5(r^2 n)(\delta^2 m) \le \frac{5\cdot \paperZcoefZrZsquared}{\paperZanticoefZmZbyZrecpZdeltasqr} d^2 = 0.8 d^2$.

    In the following, the interaction of the simulator with the input is restricted to revealing the exact value of $p_i$ of some index $i$ sampled by the algorithm (for the $m+1$st time, which is at most once per index). Consider the history $(M,I,f)$ and the queried index $i = \mathrm{ALG}(M,I,f)$.

     For a matrix $M$, let $k_i$ be the number of $1$s in the $i$th column of $M$. Observe that, once the input distribution is fixed (to either $D_n(\delta,0)$ or $D_n(\delta,r)$), the indexes are fully independent, and hence the posterior distribution of $p_i$ only depends on $k_i$.
    \begin{eqnarray*}
        \Pr_{D_n(\delta,0)}\left[p_i = \frac{1-\delta}{2} \cond M,I,f,i \right]
        &=& \Pr_{\mu}\left[p_i = \frac{1-\delta}{2} \cond k_i,i \right]
        = \frac{\frac{1}{2}\tau^{(i)}_-(k_i)}{\frac{1}{2}\tau^{(i)}_-(k_i) + \frac{1}{2}\tau^{(i)}_+(k_i)} \\
        \Pr_{D_n(\delta,r)}\left[p_i = \frac{1-\delta}{2} \cond M,I,f,i \right]
        &=& \Pr_{\nu}\left[p_i = \frac{1-\delta}{2} \cond k_i,i \right]
        = \frac{\frac{1+r}{2}\tau^{(i)}_-(k_i)}{\frac{1+r}{2}\tau^{(i)}_-(k_i) + \frac{1-r}{2}\tau^{(i)}_+(k_i)}
    \end{eqnarray*}

    From the above we deduce that:
    \begin{eqnarray*}
        \frac{\Pr_{D_n(\delta,0)}\left[p_i = \frac{1-\delta}{2} \cond M,I,f,i \right]}{\Pr_{D_n(\delta,r)}\left[p_i = \frac{1-\delta}{2} \cond M,I,f,i \right]} &=& \frac{1}{1+r}\left(1 + r \frac{\tau^{(i)}_-(k_i) - \tau^{(i)}_+(k_i)}{\tau^{(i)}_-(k_i) + \tau^{(i)}_+(k_i)}\right) \in \frac{1 \pm r}{1 + r} \subseteq [1-3r, 1+3r]
    \end{eqnarray*}

    A similar analysis shows that $\Pr_{D_n(\delta,0)}\left[p_i = \frac{1+\delta}{2} \cond M,I,f,i \right] \in (1 \pm 3r) \Pr_{D_n(\delta,r)}\left[p_i = \frac{1+\delta}{2} \cond M,I,f,i \right]$ as well. By Lemma \ref{lemma:bounded-ratio-bounded-dkl}, since $r < 1/12$,
    \[
        \E_{(M,I,f;i)} \dkl\left(\Pr_{D_n(\delta,0)}\left[p_i = \frac{1+\delta}{2} \cond M,I,f,i \right], \Pr_{D_n(\delta,r)}\left[p_i = \frac{1+\delta}{2} \cond M,I,f,i \right]\right)
        \le \frac{(3r)^2}{\ln 2}
        \le 16r^2\]
    That is, the posterior distributions of the obtained $p_i$ according to whether the input is drawn from $D_n(\delta,0)$ or from $D_n(\delta,r)$ are $16r^2$-close with respect to the KL-divergence.

    By Observation \ref{obs:markov-high-columns}, since $q \le \frac{d^2}{\paperZanticoefZqZbyZmn} mn$, there can be only 
    $\frac{d^2}{\paperZanticoefZqZbyZmn} n$ revealings. Hence, the additional KL-divergence in the second phase is bounded by $16r^2 \cdot \frac{d^2}{\paperZanticoefZqZbyZmn} n \le \frac{16d^2}{\paperZanticoefZqZbyZmn} r^2 \cdot (\paperZcoefZrZsquared/r^2) = \frac{16 \cdot \paperZcoefZrZsquared d^2}{\paperZanticoefZqZbyZmn} = 1.024 d^2$.

    Overall, The total KL-divergence of the simulator is bounded by $0.8 d^2 + 1.024 d^2 < 2d^2$. By Pinsker's inequality, the total-variation distance between the runs is bounded by $d$.
\end{proof}

\newpage
\section{Common distributions and fact sheet}
\label{apx:common-and-facts}

\subsection{Common distributions}

\begin{definition}[Bernoulli distribution]
    Bernoulli distribution $\Ber(p)$ is the distribution over $\{0,1\}$ for which $\Pr[1] = p$ and $\Pr[0] = 1-p$.
\end{definition}

\begin{definition}[Binomial distribution]
    The binomial distribution $\Bin(n,p)$ is the sum of $n$ independent samples from the Bernoulli distribution $\Ber(p)$.
\end{definition}


\begin{definition}[Normal distribution]
    The standard normal distribution over $\mathbb R$ is defined by the probability density function $\phi(x) = \frac{1}{\sqrt{2\pi}} e^{-x^2/2}$.
\end{definition}

\begin{definition}[Cumulative probability function of the normal distribution]
    The cumulative probability function of the normal distribution is denoted by $\Phi(t) = \frac{1}{\sqrt{2\pi}} \int_{-\infty}^t e^{-x^2/2} dx$.
\end{definition}

\subsection{Fact sheet}

In the following we provide some well-known bounds and identities.

\begin{lemma}[Pinsker's inequality] \label{lemma:cheatsheet:pinsker}
    $\DKL{\mu}{\tau} \ge 2(\dtv(\mu,\tau))^2$.
\end{lemma}

\begin{lemma} \label{lemma:cheatsheet:neg-ln-1+x-ge-neg-x}
    For $x > -1$, $-\ln (1+x) \ge -x$.
\end{lemma}

\begin{lemma} \label{lemma:cheatsheet:neg-ln-1+x-le-neg-x-plus-x-sqr}
    For $x \ge -2/3$, $-\ln (1+x) \le -x+x^2$.
\end{lemma}

\begin{lemma} \label{lemma:cheatsheet:dkl-of-product}
    Given distributions $\mu_1,\nu_1$ over $\Omega_1$ and $\mu_2,\nu_2$ over $\Omega_2$, $\DKL{\mu_1 \times \mu_2}{\nu_1 \times \nu_2} = \DKL{\mu_1}{\nu_1} + \DKL{\mu_2}{\nu_2}$.
\end{lemma}

\begin{lemma} \label{lemma:cheatsheet:neg-ln-1+x-le-neg-x-plus-x-sqr-half}
    For $x \ge 0$, $-\ln(1+x) \le -x + x^2/2$.
\end{lemma}

\begin{lemma} \label{lemma:cheatsheet:dkl-bin-m-p-m-q}
    $\DKL{\Bin(m,p)}{\Bin(m,q)} = m \cdot \dkl(p,q)$.
\end{lemma}

\begin{lemma}[see \cite{baynets17}] \label{lemma:cheatsheet:dkl-quadratic-bound}
    For $p,q \in [0,1]$, $\dkl(p,q) \le \frac{(p-q)^2}{q(1-q)}$.
\end{lemma}

\begin{lemma}[Chain rule for $\dkl$] \label{lemma:cheatsheet:chain-rule-dkl}
    Let $\mu$ and $\tau$ be two distributions over $A \times B$. In this setting, $\DKL{\mu}{\tau} = \DKL{\mu|_A}{\tau|_A}
    + \E_{a \sim \mu|_A}\left[\DKL{\mu|^{\{a\} \times B}}{\tau|^{\{a\} \times B}} \right]$.
\end{lemma}

\begin{lemma}[Information processing inequality] \label{lemma:cheatsheet:information-processing-inequality}
    Let $\mu$ and $\tau$ be two distribution over $\Omega$ and let $f : \Omega \to \Omega'$ be a function. Let $f(\mu)$ be the distribution in which $\Pr[y] = \Pr_\mu\left[f^{-1}(y)\right]$ and $f(\tau)$ be the distribution in which $\Pr[y] = \Pr_\tau\left[f^{-1}(y)\right]$. In this setting, $\DKL{f(\mu)}{f(\tau)} \le \DKL{\mu}{\tau}$. This also holds if $f$ is a probabilistic function (that is, for every $x \in \Omega$, $f(x)$ is a distribution over $\Omega'$ and $\Pr[y] = E_{x\sim \mu}[\Pr_{f(x)}[y]]$).
\end{lemma}

\newpage
\section{External calculations}

\begin{lemma}[Normal distribution lookup] \label{lemma:lookup:Phi(-2sqrt(3))}
    $\Phi(-2\sqrt{3}) \le e^{-8.23}$.
\end{lemma}
\begin{proof}
    WolframAlpha query: \\ln ((1/sqrt(2*pi)) * (int from (-infty) to (-2*sqrt(3)) of (exp(-x\^{}2/2)dx)))\\Answer ``$\approx -8.232$''.
    
    \url{https://www.wolframalpha.com/input?i=ln+%28%281%2Fsqrt%282*pi%29%29+*+%28int+from+%28-infty%29+to+%28-2*sqrt%283%29%29+of+%28exp%28-x%5E2%2F2%29dx%29%29%29}
\end{proof}

\begin{lemma}[Normal distribution lookup] \label{lemma:lookup:Phi(-4.9)}
    $\Phi(-4.9) \ge e^{-14.56}$
\end{lemma}
\begin{proof}
    WolframAlpha query: \\ln ((1/sqrt(2*pi)) * (int from (-infty) to (-4.9) of (exp(-x\^{}2/2)dx)))\\Answer: ``$\approx -14.5512$''.

    \url{https://www.wolframalpha.com/input?i=ln+%28%281%2Fsqrt%282*pi%29%29+*+%28int+from+%28-infty%29+to+%28-4.9%29+of+%28exp%28-x%CB%862%2F2%29dx%29%29%29}
\end{proof}

\begin{lemma}[Inequality] \label{lemma:wolfram:(1-64/n)^(n/2)exp(16sqrt(6)/(1+8/sqrt(n)))>exp(7.19)}
    For $n \ge 3 \cdot 10^{10}$, $(1 - \paperZcoefZrZsquared / n)^{\frac{1}{2}n} \cdot e^{\frac{16\sqrt{6}}{1 + 8/\sqrt{n}}} > e^{7.19}$.
\end{lemma}
\begin{proof}
    WolframAlpha query:\\
    solve (1 - 64/x)\^{}(x/2) * exp(16*sqrt(6)/(1 + 8/sqrt(x))) \textgreater~ exp(7.19) \\
    Answer: $x > 2.91646 \times 10^{10}$.

    \url{https://www.wolframalpha.com/input?i=solve+%281+-+64%2Fx%29%5E%28x%2F2%29+*+exp%2816*sqrt%286%29%2F%281+%2B+8%2Fsqrt%28x%29%29%29+%3E+exp%287.19%29}
\end{proof}

\newpage
\addcontentsline{toc}{section}{References}
\bibliographystyle{alpha}
\bibliography{main}

\end{document}